\documentclass{article}

\usepackage{fullpage}
\usepackage{amssymb}

\usepackage{lipsum}
\usepackage{amsfonts}
\usepackage{amssymb}
\usepackage{amsthm}
\usepackage{amsmath}
\usepackage{graphicx}
\usepackage{epstopdf}
\usepackage{url}
\usepackage{cite}
\usepackage{booktabs}    
\usepackage{multirow}    
\theoremstyle{plain}
\newtheorem{thrm}{Theorem}[section]
\newtheorem{lmm}[thrm]{Lemma}

\theoremstyle{definition}
\newtheorem{dfntn}[thrm]{Definition}

\newtheorem{rmrk}[thrm]{Remark}

\title{Sparse Configuration Interaction for the Electronic Schr\"odinger Equation Revisited: Complete Basis Set Limit Complexity and Quantum-Encoding Impact}
\author{Michael Griebel\thanks{Institute for Numerical Simulation, University of Bonn, Friedrich-Hirzebruch-Allee 7, 53115 Bonn, Germany.
}
\thanks{Fraunhofer Institute for Algorithms and Scientific Computing, Schloss Birlinghoven, D-53757 Sankt Augustin, Germany. ({\em jan.hamaekers@scai.fraunhofer.de})
  }
\and Jan Hamaekers\footnotemark[2]
}

\usepackage{amsopn}

\newcommand\dvec[1]{\mathbf{#1}}
\newcommand\Nvec[1]{\vec{#1}}
\newcommand\Ndvec[1]{\mspace{3mu}\Nvec{\dvec{#1}}\mspace{3mu}}
\newcommand\nfrac[2]{#1/#2}
\newcommand\Intd[1]{\,d{#1}}

\newcommand{\Oh}{\mathcal{O}}

\newcommand\myell{L}

\newcommand\spinu{\uparrow}
\newcommand\spind{\downarrow}
\newcommand\spinup{{+\tfrac{1}{2}}}
\newcommand\spindown{{-\tfrac{1}{2}}}

\newcommand\orderO{{\mathcal{O}}}
\newcommand\myiota{{\xi}}
\newcommand\myKa{{|\mathcal{B}^1_L|}}
\newcommand\myKb{{|\mathcal{B}^1_L|}}
\newcommand\myKc{{|\mathcal{B}^1_L|}}
\newcommand\mytfrac[2]{{#1}/{\left(#2\right)}}
\newcommand\myg{{\tilde{f}}}

\begin{document}

\maketitle

\begin{abstract}
In this article we revisit regularity results for eigenfunctions in the discrete spectrum of the electronic Schr\"odinger equation and study their consequences for approximation complexity.
In particular, for the convergence to the complete basis set limit, it can be shown that the curse of dimensionality in the leading algebraic exponent can be mitigated.
That is, for general sparse grid constructions, the main term of the convergence rate with respect to the number of degrees of freedom is independent of the number of electrons.
These insights indicate potential benefits for classical numerical solvers of the electronic Schr\"odinger equation and also for quantum-computing approaches through new qubit-efficient wavefunction encodings.
\end{abstract}



\section{Introduction}
 One important goal of computational quantum chemistry is the solution of the Schr\"odinger equation, since it describes the behavior of atoms and electrons on the nanoscale. Due to the high dimensionality of the wavefunction, this task is regarded as a major challenge because of the {\em curse of dimensionality} (CoD). As will be discussed below, for conventional discretizations, the asymptotic convergence rate with respect to the number of involved degrees of freedom may deteriorate with dimension, here $3N$ for $N$ electrons. In addition, the constants entering the error estimates may also grow rapidly with dimension. Quantum-computing based algorithms may constitute a possible future approach to these difficulties, provided that sufficiently capable hardware becomes available.
 But as will be discussed in this article, in the case of the isolated eigenvalues of the discrete spectrum and associated eigenfunctions of the electronic Schr\"odinger operator, certain regularity properties help, on the one hand, to circumvent the CoD also for conventional approaches on classical computing resources, at least to some extent, and, on the other hand, to reduce the number of qubits required to encode such wavefunctions in the context of quantum computing algorithms.

The reference method for the solution of the electronic Schrödinger equation is the so-called Full Configuration Interaction (FCI) approach \cite{helgaker2013molecular}.
In FCI, a finite set of appropriate one-particle functions is used to build $N$-electron Slater determinants as the basis to discretize the eigenvalue problem within a conventional Galerkin approach.
Here, if $\mathcal{B}_\infty^{1}$ is a complete one-particle basis set, e.g.\ the Gauss-Hermite functions in $\mathbb{R}^3$, and if $\{\mathcal{B}^{1}_L\}_{L\in\mathbb{N}}$ is a nested sequence of finite subsets of $\mathcal{B}_\infty^{1}$, i.e.
\begin{equation*}
    \mathcal{B}^{1}_L \subset \mathcal{B}^{1}_{L+1} \text{ for all } L\in\mathbb{N} \text{  and } \bigcup_{L\in\mathbb{N}} \mathcal{B}^{1}_L = \mathcal{B}_\infty^{1}
    ,
\end{equation*}
then, for each finite one-particle basis function set $\mathcal{B}^{1}_L$, the constructed $N$-electron Slater determinants form a finite basis set
    $\mathcal{B}^{FCI}_{\mathcal{B}^1_L}$
for an associated finite-dimensional subspace of the $N$-electron space, as we define later in (\ref{equ:FCIbs}). Moreover, the FCI solution $\Psi_{\mathcal{B}^{FCI}_L}$ is the Galerkin approximation of the exact solution in this finite-dimensional subspace. This way, the sequence of FCI solutions $\Psi_{\mathcal{B}^{FCI}_L}$ converges to the exact solution in the so-called complete basis set limit (CBS) for $L\to \infty$, while the number $M$ of $N$-electron basis functions is in the order of
\begin{equation*}
    M := |\mathcal{B}^{FCI}_{\mathcal{B}^1_L}| = \binom{\myKa}{N} = \orderO\left(\myKa^N\right).
\end{equation*}
Here, we neglect spin up and down for reasons of simplicity.
Note that if we just use that the eigenfunctions are at least in $\mathcal{H}^1((\mathbb{R}^{3})^{N})$, we encounter the curse of dimensionality, i.e.\ the rate
of convergence for the CBS of the FCI method deteriorates with the number $N$ of involved electrons.
Moreover, if we consider the case of an increasing system size instead of a fixed number of electrons and assume that the number $\myKa$ of one-particle basis functions is proportional to the number $N$ of involved electrons,
then the FCI approach leads to exponentially increasing costs with the number $N$ of electrons involved, since the number of corresponding Slater determinants is given by $\binom{\myKa \approx N}{N} \approx \orderO(N^N)$, where we neglected spin again.

To overcome the curse of dimensionality of the FCI approach, various model approximations have been developed and
implemented in quantum chemistry which lead to reduced complexity.
Prominent examples are Hartree-Fock and multiconfiguration approaches (HF/MCHF), additive expansions such as configuration interaction, exponential methods such as coupled cluster techniques, perturbative schemes like M\"oller-Plesset theory, quantum Monte Carlo methods (VMC/DMC), and density-functional theory (DFT).
For some of these approaches, rigorous existence and error analyses are available, see e.g.~\cite{friesecke2003multiconfiguration,rohwedder2013error}. However, it is in general not yet completely clear how to systematically improve their approximation error with guaranteed convergence rates in a mathematically rigorous way.

Alternatively, the pioneering work of Yserentant \cite{yserentant2004regularity,yserentant2005sparse} opened the way for approaches~\cite{griebel2010tensor,hamaekers2009tensor,bachmayr2012hyperbolic,zeiser2012wavelet,kogure2022wavelet} which, on the one hand, systematically converge to the exact eigenfunctions, i.e. to the CBS of FCI, and on the other hand, are able to mitigate the curse of dimensionality in the main algebraic rate.
This is due to the fact that the eigenfunctions belong to certain Sobolev spaces of dominating mixed smoothness. Such spaces with dominating mixed smoothness are directly related to sparse grid approximation schemes, i.e.\ Sparse Configuration Interaction (SCI) approaches that promise to avoid the CoD, at least in the main term of the convergence rate~\cite{Bungartz.Griebel:2004}.
In particular, a related SCI approach was discussed in~\cite{anderson2018breaking}, which is based on general sparse grid approximation spaces introduced in~\cite{Griebel2000,Griebel2008}.

In this work we discuss a general sparse grid based approach for Sobolev spaces of dominating mixed smoothness with exponential spatial decay. This way, we improve some of our previous results~\cite{griebel2010tensor,hamaekers2009tensor}. In addition, we discuss its implications for the complete basis set limit and for possible encodings with a reduced number of qubits in quantum computing algorithms.
Compared to our earlier approach in \cite{griebel2010tensor,hamaekers2009tensor}, the present work has the advantage of a unified framework based on {\em weighted} spaces with simultaneous mixed smoothness and spatial decay that results from a Meyer wavelet characterization. A key ingredient in this framework is the use of a smooth Meyer wavelet family (cf. Appendix \ref{app:MeyerWaveletFamily}), which provides the required localization and weighted norm-characterization properties. Moreover, it gives a general approximation principle which separates level and translation truncation effects. This yields error bounds, degree of freedom estimates, and $\varepsilon$-complexity statements in one consistent setting. In particular, the final convergence statements retain the same main algebraic rate in
the involved number of degrees of freedom as in the sparse grid/hyperbolic cross analysis discussed in \cite{griebel2010tensor,hamaekers2009tensor}, and thus no $N$-dependence occurs in the main algebraic exponent while the $N$-dependency of the logarithmic terms and constants becomes explicit.
In relation to wavelet based sparse constructions as studied, for example, in \cite{zeiser2012wavelet,kogure2022wavelet}, the present analysis is tailored to the electronic Schr\"odinger setting with mixed Sobolev regularity and exponential decay and makes the implications for CBS convergence and qubit scaling explicit.

The remainder of this paper is organized as follows. In Section \ref{sec:RegularitySchroed} we recall the electronic Schrödinger equation and some properties of its solution. In Section \ref{sec:SparseCI} we introduce many-particle Sobolev spaces with spatial decay and present an approach for the efficient approximation of functions from these spaces including associated upper bounds for the error and the involved degrees of freedom. In Section \ref{sec:GSGApproxWaveFunction} we apply our results to eigenfunctions in the discrete spectrum of the electronic Schr\"odinger equation. Next, in Section \ref{sec:limits}, we discuss the impact of our results on different limits associated with approximation methods for the electronic Schr\"odinger equation. In Section \ref{sec:QuantumComputing}, we discuss implications for quantum computing algorithms, including qubit-scaling comparisons for different encodings and implementation-level caveats. We finally draw some conclusions in Section \ref{sec:conclusions}.

 \section{The Electronic Schrödinger Equation}\label{sec:RegularitySchroed}
For most problems in quantum chemistry, it is sufficient to consider the so-called Born-Oppenheimer approximation of the electronic Schr\"odinger equation, which
neglects relativistic effects, treats nuclei as classical point-like particles, and employs just the electrons as quantum mechanical particles.
For $N_{\text{nuc}}$ trapped nuclei $(\dvec{R}_k, Z_k)$ and $N$ electrons with space and spin coordinates $(\dvec{x}_i, s_i)\in\mathbb{R}^{3}\times \{\spindown,\spinup\}$, it reads as the eigenvalue problem
\begin{equation}\label{equ:SE}
    H \Psi = E\Psi
\end{equation}
with the spin-independent molecular Hamiltonian
\begin{equation*}
  H = - \frac{1}{2}\sum_{i=1}^N \Delta_i
  - \sum_{k=1}^{N_{\text{nuc}}} \sum_{i=1}^N \frac{Z_k}{|\dvec{x}_i-\dvec{R}_k|}
  + \sum_{i<j}^N \frac{1}{|\dvec{x}_i-\dvec{x}_j|}
  + \sum_{k<l}^{N_{\text{nuc}}} \frac{Z_k Z_l}{|\dvec{R}_k-\dvec{R}_l|}
\end{equation*}
written in atomic units. The solution to the $N$-electron problem (\ref{equ:SE}) lives in a $3N$-dimensional space and respects the antisymmetry condition imposed by Pauli's principle, i.e.\
\begin{equation}\label{equ:Antisymmetry}
\Psi((\dvec{x}_{P(1)},s_{P(1)}), \ldots, (\dvec{x}_{P(N)},s_{P(N)})) = (-1)^{|P|} \Psi((\dvec{x}_1,s_1), \ldots, (\dvec{x}_N,s_N))
\end{equation}
for all permutations $P$ of the symmetric group $\mathcal{S}_N$. Thus the spin variables $s_1, \ldots, s_N$ enter only through the antisymmetry conditions (\ref{equ:Antisymmetry}). Furthermore, for a given spin distribution $\Nvec{s} \in \{\spindown,\spinup\}^{N}$, the spatial part
$\Psi_{\Nvec{s}}(\Ndvec{x}):=\Psi(\Ndvec{x},\Nvec{s})$ obeys the
partial antisymmetry condition
\begin{equation*}
  \Psi_{\Nvec{s}}(P\Ndvec{x}) = (-1)^{|P|}\Psi_{\Nvec{s}}(\Ndvec{x}),
  \quad \forall P\in \mathcal{S}_{\Nvec{s}}:= \left\{P\in\mathcal{S}_N \, :
    \, P\Nvec{s}=\Nvec{s} \right\}
\end{equation*}
and problem (\ref{equ:SE}) is reduced to the eigenvalue problems
\begin{equation}\label{equ:SEspin}
    H \Psi_{\Nvec{s}} = E_{\Nvec{s}}\Psi_{\Nvec{s}},
\end{equation}
where $\Psi_{\Nvec{s}}: \mathbb{R}^{3N} \mapsto \mathbb{R}$ obeys the above partial antisymmetry condition. In the case of a spin-independent Hamiltonian $H$ with $N$ electrons, it is sufficient to consider only the $\lfloor\frac{N}{2}+1\rfloor$ spin distributions
\begin{equation*}
  \Nvec{s}^{\,(0,N)}, \ldots, \Nvec{s}^{\,(\lfloor\frac{N}{2}\rfloor,N)}, \quad \text{ where }
    \Nvec{s}^{\,(i,N)} := (\underbrace{\spindown, \ldots, \spindown}_{\text{$i$ times}}, \underbrace{\spinup, \ldots, \spinup}_{\text{$N-i$ times}}),
\end{equation*}
as discussed in \cite{yserentant2010regularity}.

The main numerical challenge in directly solving eigenvalue problem (\ref{equ:SEspin}) is its high dimensionality.
A conventional Galerkin method for discretizing the problem using a set of $N$-particle Slater determinants $\mathcal{B}^{FCI}_{\mathcal{B}^1_L}$ built from a finite basis set $\mathcal{B}^1_L$ of one-particle functions,
as in the conventional FCI method \cite{helgaker2013molecular} by
\begin{equation}\label{equ:FCIbs}
    \mathcal{B}^{FCI}_{\mathcal{B}^1_L}(N_\spinu, N_\spind) := \left\{ \bigwedge_{i=1}^{N_\spind} \phi_{k^{\spind}_i} \otimes \bigwedge_{j=1}^{N_\spinu} \phi_{k^{\spinu}_j}
    \, : \, 1 \leq k^{\spind}_1 < \cdots < k^{\spind}_{N_\spind} \leq \myKb \, \text{ and }
    1 \leq k^{\spinu}_1 < \cdots < k^{\spinu}_{N_\spinu} \leq \myKb
    \right\},
\end{equation}
would involve
\begin{equation*}
  |\mathcal{B}^{FCI}_{\mathcal{B}^1_L}(N_\spinu, N_\spind)| = \binom{\myKb}{N_\spind}\binom{\myKb}{N_\spinu} \approx  \orderO\left(\myKa^{N_\spinu + N_\spind}\right)
\end{equation*} degrees of freedom.
Note here that the error of the eigenvalue approximations behaves like the square of the corresponding eigenfunction error in the conventional Sobolev norm in $\mathcal{H}^1((\mathbb{R}^{3})^{N})$ \cite{yserentant2010regularity}, i.e.
\begin{equation*}
 \left|E_{\Nvec{s}}-E_{\mathcal{B}^{FCI}_{\mathcal{B}^1_L}}\right| \leq C(N) \|\Psi_{\Nvec{s}} -\Psi_{\mathcal{B}^{FCI}_{\mathcal{B}^1_L}}\|_{\mathcal{H}^1}^2,
\end{equation*}
where $C$ is a constant that may depend on $N$ but not on $L$.

Let us now choose Gauss-Hermite functions as the underlying one-particle basis set and assume that an eigenfunction $\Psi_{\Nvec{s}}$ is at least in the isotropic Sobolev space $\mathcal{H}^{r+1}((\mathbb{R}^{3})^{N})$. Then, from standard spectral analysis results \cite{shen2009some}
only a convergence estimate of type
\begin{equation*}
  \|\Psi_{\Nvec{s}} -\Psi_{\mathcal{B}^{FCI}_{\mathcal{B}^1_L}}\|_{\mathcal{H}^1} \leq c(N) M^{-\tfrac{1}{2}\tfrac{r}{3N}}\|\Psi_{\Nvec{s}}\|_{\mathcal{H}^{r+1}}
\end{equation*}
can be derived.
Here, we encounter the curse of dimensionality, i.e.\ the rate
of convergence for the complete basis set limit (CBS) of the FCI method deteriorates with the number $N$ of involved electrons. Additionally note that the constant $c(N)$ and also the wavefunction norm
$\|\Psi_{\Nvec{s}}\|_{\mathcal{H}^{r+1}}$ are not dependent on the number of degrees of freedom $M$, but may still grow strongly with the number of electrons $N$ when considering sequences of systems with increasing electron number.

As outlined already in the introduction,
in case of an increasing system size instead of a fixed number of electrons,
the FCI approach leads to exponentially increasing costs with the number $N$ of electrons involved, e.g.\ in case of total antisymmetry (all electrons have the same spin) the number $M$ of involved Slater-determinants is given by $M=\binom{\myKa \approx N}{N} \approx \orderO(N^N)$.
To mitigate this unfavorable scaling of FCI, a wide range of approximate electronic-structure methods have been developed in quantum chemistry.
However, despite their practical success, a fully rigorous framework for systematic error improvement is still missing.
Here, the regularity theory developed by Yserentant \cite{yserentant2004regularity,yserentant2005sparse} opened the possibility for methods that preserve the convergence to the CBS of the FCI ansatz while mitigating dimensional bottlenecks.
Concretely, it was shown that the $N$-th order mixed derivatives of the eigenfunctions of bounded states, i.e. the eigenfunctions of the discrete spectrum, both exist and decay exponentially, see~\cite{yserentant2010regularity,yserentant2011mixed}.
In more detail, Yserentant proved that the bounded states are in the anisotropic Sobolev space $\mathcal{H}_{mix}^{t,1}$ with $t<\tfrac{3}{4}$ for the general spin distribution case and with $t=1$ for the special case of total antisymmetry. Recently, it has been shown in~\cite{meng2023mixed} that the result for the total antisymmetric case can be improved, such that even $t<\tfrac{5}{4}$ holds.
Altogether, the smoothness of the wavefunction increases with the number of electrons
such that the main rate of convergence does not deteriorate with the number of electrons. Namely,
in \cite{kreusler2012mixed,yserentant2010regularity} appropriate
sparse CI $N$-particle basis sets $\mathcal{B}^{SCI}_{\mathcal{B}^1_L}$ have been constructed and associated convergence estimates of type
\begin{equation*}
  \|\Psi_{\Nvec{s}} -\Psi_{\mathcal{B}^{SCI}_{\mathcal{B}^1_L}}\|_{\mathcal{H}^1} \leq c(N)  M^{-\tfrac{1}{1+\alpha}\tfrac{t}{3}} m(M,N) \|\Psi_{\Nvec{s}}\|_{\mathcal{H}^{t,1}_{mix}}
\end{equation*}
have been derived
using eigenfunctions of certain one-particle Schrödinger like operators as the underlying one-particle basis set $\mathcal{B}^1_L$, where $M:=|\mathcal{B}^{SCI}_{\mathcal{B}^1_L}|$ is
the number of Slater determinants in the sparse CI basis $\mathcal{B}^{SCI}_{\mathcal{B}^1_L}$,
$m$ is some logarithmic term in $M$ depending on the specific method (like e.g.\  $m(M,N)=\log(M)^{N-1}$) and $\alpha$ is dependent on the underlying one-particle basis set.
For example, $\alpha=1$ for Gauss-Hermite functions resulting from the one-particle harmonic oscillator and it holds in general $\alpha>0$ as shown in \cite{yserentant2010regularity}. Moreover, alternative underlying one-particle basis sets and sparse grid or hyperbolic cross based approximation space constructions have been discussed in \cite{griebel2007sparse,griebel2010tensor,hamaekers2009tensor,zeiser2012wavelet}. We discuss the construction of our specific sparse CI basis set $\mathcal{B}^{SCI}_{\mathcal{B}^1_L}$ in Section~\ref{sec:limits}.

In \cite{griebel2010tensor,hamaekers2009tensor} a tensor product multiscale basis has been used to show for arbitrary $\delta >0$ a main rate\footnote{Here, we neglect the logarithmic terms.} of $\orderO(M^{-\mytfrac{t}{3(1+\delta)}})$ for the error in the $\mathcal{H}^1$-norm for functions in anisotropic Sobolev spaces $\mathcal{H}_{mix}^{t,1}((\mathbb{R}^{3})^N)$ with certain spatial decay properties,
where in particular the case $t=\tfrac{1}{2}$ has been discussed. Furthermore, the general partial antisymmetric case $t=\tfrac{3}{4(1+\epsilon)}$, $\epsilon >0$ results in a main rate of $\orderO(M^{-\mytfrac{1}{4(1+\epsilon)(1+\delta)}})$ for the wavefunction
error in the $\mathcal{H}^1$-norm and hence $\orderO(M^{-\mytfrac{1}{2(1+\epsilon)(1+\delta)}})$ for the eigenvalue error. Accordingly, the total antisymmetric case $t=\frac{5}{4(1+\epsilon)}$ results in a main rate of $\orderO(M^{-\mytfrac{5}{12(1+\epsilon)(1+\delta)}})$ for the wavefunction error and $\orderO(M^{-\mytfrac{5}{6(1+\epsilon)(1+\delta)}})$ for the eigenvalue error. In addition, in \cite{kreusler2012mixed} it has been shown that the eigenfunctions can be written as the product of a universal factor that covers the electron-electron cusps and a regular part $\Phi \in \mathcal{H}_{mix}^{t,1}$ with $t=1$ independent of whether $\Psi$ obeys a total antisymmetric or just a general partial antisymmetric spin distribution. Accordingly, using a sparse grid/hyperbolic cross type approach \cite{griebel2010tensor,hamaekers2009tensor}, a main rate of $\orderO(M^{-\mytfrac{1}{3(1+\delta)}})$ for the error of $\Phi$ in the $\mathcal{H}^1((\mathbb{R}^3)^N)$-norm can be deduced, which corresponds to the general isotropic one-particle case.

Moreover, several wavefunction based methods can deliver near benchmark-level accuracy for suitable problem classes.
In particular, coupled-cluster hierarchies (e.g.\ CCSD(T) and higher-order variants) are among the most successful systematic correlation methods in weakly correlated regimes, often yielding highly accurate energies in practice \cite{bartlett2007coupled,helgaker2013molecular}.
For strongly correlated systems, tensor-network representations such as matrix-product states, the density-matrix renormalization group (DMRG), and related tensor-train formulations, provide controlled low-rank approximations in very large active spaces \cite{schollwock2011dmrg}.
In addition, selected and incremental configuration-interaction strategies (including modern many-body expanded FCI formulations) offer another route to high-accuracy solutions and have been developed extensively, see e.g.\ \cite{eriksen2021incremental}.
Complementary stochastic CI methods, most prominently the full configuration-interaction quantum Monte Carlo (FCIQMC) approach, provide a probabilistic route to near-FCI accuracy in large determinant spaces \cite{booth2009fciqmc}.

Let us note furthermore that, if appropriate quantum computing devices will become available in the future, a further approach to tackle the curse of dimensionality would be to apply quantum computing algorithms \cite{babbush2016exponentially,babbush2017exponentially,lee2023evaluating,su2021fault}.
One crucial aspect of these algorithms is the encoding of wavefunctions onto a finite set of qubits, where, for example, fermionic mapping schemes are employed in the second quantization framework.
When using $\myKc$ spin orbitals to approximate the $N$-electron wavefunction, the corresponding quantum states can be represented using $n = O(\myKc)$ qubits. Additionally, there exist efficient and compact encoding schemes that further reduce the required qubit count. For fixed numbers $N_\spinu$ and $N_\spind$ of spin-up and spin-down electrons, with $N = N_\spinu + N_\spind$, the quantum states — and similarly, the number of Slater determinants in first quantization — can be mapped to just $\left\lceil\log_2\left(\binom{\myKc}{N_\spinu}\binom{\myKc}{N_\spind}\right)\right\rceil$ qubits. For more details, see \cite{babbush2017exponentially,chamaki2022compacct}.

In this paper, we pursue a deterministic sparse CI approach based on the available mixed-regularity theory and thereby obtain rigorous a priori error estimates together with corresponding $\varepsilon$-complexity bounds. Our results imply a substantial reduction in the number of necessary Slater determinants and, hence, in the number of degrees of freedom, which translates into significantly reduced computational costs for conventional Galerkin discretization approaches on classical computers. At the same time, these reductions are also relevant in the context of quantum algorithms for the electronic Schr\"odinger equation, because state-space compression directly affects qubit requirements.

\section{Approximation in Many-Particle Sobolev Spaces}\label{sec:SparseCI}
 In this section we introduce an approach to approximate functions in many-particle Sobolev spaces of dominating mixed smoothness with spatial decay toward infinity by appropriate Meyer wavelet based general sparse grid spaces.
To this end, after a few preliminaries, we first define the necessary many-particle spaces and discuss their characterization by Meyer wavelets. Then, we introduce and analyze appropriate general sparse grid approximation spaces and their application to approximate functions in many-particle Sobolev spaces of dominating mixed smoothness on $(\mathbb{R}^D)^N$ with an {\em exponential spatial decay} toward infinity.

Let us note here that, in comparison to the approximation on a compact domain such as the unit cube or torus, approximation on the whole space $(\mathbb{R}^D)^N$ requires additional control of the behavior toward $\pm$ infinity. On $([0,1]^D)^N$, the main issues are mixed regularity and possible boundary effects, whereas on the whole space $(\mathbb{R}^D)^N$ one has to balance local resolution of singular or anisotropic features with global spatial decay in weighted norms. For sparse grid/hyperbolic cross and wavelet based constructions this means that,
besides the usual level selection in mixed directions, one needs an additional mechanism to also control the selection of finite subsets of the infinitely many translates that is based on the decay of the coefficients induced by the spatial decay. This interplay between mixed smoothness and spatial decay is the key reason why whole-space approximation needs a different approximation-space design than its compact-domain counterpart.

\subsection{Preliminaries}
In this work, we consider various types of many-particle Sobolev spaces associated with $N$ particles in dimension $D$. Here, we use the notation $\Ndvec{x} = (\dvec{x}_1, \ldots, \dvec{x}_N) \in (\mathbb{R}^D)^N$ with $\dvec{x}_i = (x_{i,1}, \ldots, x_{i,D}) \in \mathbb{R}^D$ for the involved variables and we define a norm on $(\mathbb{R}^D)^N$ by
\begin{equation*}
    |\Ndvec{k}|_{p,q} :=  |(|\dvec{k}_1|_p, \ldots, |\dvec{k}_N|_p)|_q = \left(\sum_{i=1,\ldots,N} |\dvec{k}_i|_p^q\right)^{\tfrac{1}{q}},
\end{equation*}
 where $|\cdot|_p$ is the conventional $p$-norm on $\mathbb{R}^D$.
For
$N$-dimensional multi-indices, we use the notation $\Nvec{l} \in \mathbb{N}^N$.
We denote the set $\{1, 2, 3, \ldots \}$ of positive integers by $\mathbb{N}$ and the set $\{0, 1, 2, \ldots \}$ of natural numbers including zero by $\mathbb{N}_0$.
Moreover, we define the Fourier transform of an absolutely integrable
function $f \in \mathcal{L}^1((\mathbb{R}^D)^N)$ by
\begin{equation}\label{equ:FourierTransform}
  \mathcal{F}[f](\Ndvec{k}) := \left(\frac{1}{\sqrt{2\pi}}\right)^{DN}\int_{(\mathbb{R}^D)^N} e^{-i \Ndvec{k}^T \Ndvec{x}} f(\Ndvec{x})  \Intd{\Ndvec{x}}
  , \quad \Ndvec{k} \in (\mathbb{R}^D)^N
   ,
\end{equation}
where we also write $\hat{f} := \mathcal{F}[f]$. Note here that
$\Ndvec{k}^T \Ndvec{x} = \sum_{j=1}^N \dvec{k}_j^T \dvec{x}_j$ is equal to the scalar
product on $(\mathbb{R}^D)^N$. We denote the Schwartz space of all
complex-valued rapidly decreasing infinitely differentiable functions
in $(\mathbb{R}^D)^N$ by
\begin{equation*}
  \mathcal{S}((\mathbb{R}^D)^N) = \left\{ f \in C^\infty((\mathbb{R}^D)^N) \,:\, \sup_{\Ndvec{x}\in(\mathbb{R}^D)^N} |\Ndvec{x}^{\Ndvec{\beta}} D^{\Ndvec{\alpha}} f| < \infty,\, \Ndvec{\alpha}, \Ndvec{\beta} \in (\mathbb{N}_0^{D})^N \right\}
\end{equation*}
and its topological dual, the space of all tempered distributions on
$(\mathbb{R}^D)^N$, by $\mathcal{S}^\prime((\mathbb{R}^D)^N)$. Then, for $f
\in \mathcal{S}((\mathbb{R}^D)^N)$, the Fourier transform $\hat{f}$ given
by (\ref{equ:FourierTransform}) is again in
$\mathcal{S}((\mathbb{R}^D)^N)$ and in particular the mapping $\mathcal{F}
: \mathcal{S}((\mathbb{R}^D)^N) \rightarrow \mathcal{S}((\mathbb{R}^D)^N)$ is
a linear isomorphism on $\mathcal{S}((\mathbb{R}^D)^N)$. Moreover, the inverse
Fourier transform operator $\mathcal{F}^{-1}$ is given by
\begin{equation}
  \mathcal{F}^{-1}[\hat{f}](\Ndvec{x})  = \left(\frac{1}{\sqrt{2\pi}}\right)^{DN}\int_{(\mathbb{R}^D)^N} e^{i \Ndvec{x}^T \Ndvec{k}} \hat{f}(\Ndvec{k})  \Intd{\Ndvec{k}}
  .
\end{equation}

It is well-known that both $\mathcal{F}$ and
$\mathcal{F}^{-1}$ can be extended to a continuous, linear and
bijective operator from $\mathcal{S}^\prime((\mathbb{R}^D)^N)$ onto
itself. Then, the Fourier transform $\mathcal{F}[f]$ of
a tempered distribution $f \in \mathcal{S}^\prime((\mathbb{R}^D)^N)$ can
be defined by
\begin{equation*}
  \langle \mathcal{F}[f] \,,\, g \rangle = \langle f \,,\, \mathcal{F}[g] \rangle \quad \text{ for all } g \in \mathcal{S}((\mathbb{R}^D)^N)
  .
\end{equation*}
Note that $\mathcal{S} \subset \mathcal{L}^2 \subset
\mathcal{S}^\prime$ and $\mathcal{S}$ is dense in
$\mathcal{L}^2$. Moreover, $f \in \mathcal{L}^2((\mathbb{R}^D)^N)$ implies
$\hat{f} \in \mathcal{L}^2((\mathbb{R}^D)^N)$ and the Fourier transform
operator $\mathcal{F}$ is a unitary isometric automorphism of
$\mathcal{L}^2((\mathbb{R}^D)^N)$. Here, for $f,g \in
\mathcal{L}^2((\mathbb{R}^D)^N)$, the Fourier transforms $\hat{f},\hat{g}
\in \mathcal{L}^2((\mathbb{R}^D)^N)$ preserve the
$\mathcal{L}^2$-inner-product, i.e.\
\begin{equation}\label{equ:ParsevalPlancherel}
  \langle f \,,\, g \rangle_{\mathcal{L}^2} = \langle \hat{f} \,,\, \hat{g} \rangle_{\mathcal{L}^2}
  .
\end{equation}
Thus, $\mathcal{F}$ is $\mathcal{L}^2$-norm conserving,
i.e.\ $\|f\|_{\mathcal{L}^2} = \|\hat{f}\|_{\mathcal{L}^2}$ for $f \in
\mathcal{L}^2((\mathbb{R}^D)^N)$. Also, analogously to the so-called
Parseval (or Plancherel) formula (\ref{equ:ParsevalPlancherel}), the
identity $\langle f \,,\, g \rangle = \langle \hat{f} \,,\, \hat{g}
\rangle$ holds for $f \in \mathcal{S}^\prime((\mathbb{R}^D)^N)$ and $g \in
\mathcal{S} ((\mathbb{R}^D)^N)$.

Finally we recall some well-known relations concerning the regularity and the
decay of a function and its Fourier transform, where we consider the case of $f \in
\mathcal{L}^2((\mathbb{R}^D)^N)$. First, the function $f$ has weak partial derivatives
$D^{\Ndvec{\alpha}}_{\Ndvec{x}}f$ for all $|\Ndvec{\alpha}|_{1,1} \leq l$, if and only if
$\int_{(\mathbb{R}^D)^N} \left(\sum_{p=1}^{N} |\dvec{k}_p|_2^{2}\right)^l|\hat{f}(\Ndvec{k})|^2 \Intd{\Ndvec{k}} <
\infty$.
Second, if
$\int_{(\mathbb{R}^D)^N}\left(\sum_{p=1}^{N} |\dvec{x}_p|_2^2\right)^l|f(\Ndvec{x})|^2 \Intd{\Ndvec{x}} <
\infty$,
then its Fourier transform $\hat{f}$ possesses weak partial
derivatives $D^{\Ndvec{\alpha}}_{\Ndvec{k}}\hat{f}$ for all
$|\Ndvec{\alpha}|_{1,1} \leq l$.

\subsection{Many-particle Sobolev spaces with spatial decay}
We first consider ($\mathcal{L}^2$ based) Sobolev spaces defined with the help of an appropriate weight function $w : (\mathbb{R}^D)^N \rightarrow \mathbb{R}$.
\begin{dfntn}\label{def:wSobolevParticles}
Let $w : (\mathbb{R}^D)^N \rightarrow \mathbb{R}$ be an appropriate positive continuous weight function. For $D \in \mathbb{N}$ and $N \in \mathbb{N}$ we set
  \begin{equation*}
    \mathcal{H}_{w}((\mathbb{R}^D)^N) := \left\{ f \in \mathcal{S}^\prime((\mathbb{R}^D)^N) \, : \, w \hat{f} \in \mathcal{L}^2((\mathbb{R}^D)^N) \right\}
  \end{equation*}
  with the norm
  \begin{equation*}
    \| f \|_{\mathcal{H}_{w}} := \sqrt{\int_{(\mathbb{R}^D)^N} (w(\Ndvec{k}))^2|\hat{f}(\Ndvec{k})|^2 \Intd{\Ndvec{k}}}
    , \quad \Ndvec{k} \in (\mathbb{R}^D)^N.
  \end{equation*}
\end{dfntn}
This way the weight $w$ implicitly expresses some smoothness properties of the corresponding function class.
In particular, the Sobolev space
$\mathcal{H}_w((\mathbb{R}^D)^N)$ is a Hilbert space with the Hermitian inner product
\begin{equation*}
    \langle f \,,\, g \rangle_{\mathcal{H}_{w}}
    :=  \langle w \hat{f} \,,\, w \hat{g} \rangle_{\mathcal{L}^2}
    = \int_{(\mathbb{R}^D)^N} (w(\Ndvec{k}))^{2} \hat{f}^*(\Ndvec{k}) \hat{g}(\Ndvec{k}) \Intd{\Ndvec{k}}
    .
\end{equation*}

In the following we work with the algebraic isotropic weight function
\begin{equation}\label{equ:HyperbolicCrossTwiso}
  w_{\text{iso}} : (\mathbb{R}^D)^N \rightarrow \mathbb{R} : \Ndvec{k} \mapsto \sqrt{1+\max_{n=1,\ldots,N}{|\dvec{k}_n|^2_{\infty}}},
\end{equation}
the algebraic anisotropic mixed weight function
\begin{equation}\label{equ:HyperbolicCrossTwmix}
  w_{\text{mix}} : (\mathbb{R}^D)^N \rightarrow \mathbb{R} : \Ndvec{k} \mapsto \sqrt{\prod_{p=1}^N(1+|\dvec{k}_p|_{\infty}^2)},
\end{equation}
and the exponential weight function
\begin{equation}\label{equ:Expw}
  w_{\text{exp}} : (\mathbb{R}^D)^N \rightarrow \mathbb{R} : \Ndvec{k} \mapsto e^{|\Ndvec{k}|_{\infty,1}}.
\end{equation}

Many-particle Sobolev spaces of general dominating mixed smoothness \cite{griebel2010tensor,hamaekers2009tensor} can then be defined with the help of the product weight function
$w_{\text{mix}}^t w_{\text{iso}}^r$ as follows:
\begin{dfntn}
  For $-\infty < t, r < \infty$, $D \in \mathbb{N}$, $N \in \mathbb{N}$ we set
  \begin{equation*}
    \mathcal{H}^{t,r}_{\text{\emph{mix}}}((\mathbb{R}^D)^N) := \mathcal{H}_{w_{\text{\emph{mix}}}^t w_{\text{\emph{iso}}}^r}((\mathbb{R}^D)^N).
  \end{equation*}
  with the corresponding norm
  \begin{equation*}
    \| f \|_{\mathcal{H}^{t,r}_{\text{\emph{mix}}}} :=  \| f \|_{\mathcal{H}_{w_{\text{\emph{mix}}}^t w_{\text{\emph{iso}}}^r}}
    .
  \end{equation*}
\end{dfntn}
Analogously, many-particle Sobolev spaces with exponential weights can be defined with the help of the weight function $w_{\text{exp}}$ by:
\begin{dfntn}
  With (\ref{equ:Expw}) and for $s < \infty$, $D \in \mathbb{N}$, $N \in \mathbb{N}$ we set
  \begin{equation*}
    \mathcal{H}^{s}_{\text{\emph{exp}}}((\mathbb{R}^D)^N) := \mathcal{H}_{w_{\text{\emph{exp}}}^s} ((\mathbb{R}^D)^N),
  \end{equation*}
  with the corresponding norm
  \begin{equation*}
    \| f \|_{\mathcal{H}^{s}_{\text{\emph{exp}}}} :=  \| f \|_{\mathcal{H}_{w_{\text{\emph{exp}}}^s}} = \sqrt{\int_{(\mathbb{R}^D)^N}
    \left(w_{\text{\emph{exp}}}(\Ndvec{k})\right)^{2s}|\hat{f}(\Ndvec{k})|^2 \Intd{\Ndvec{k}}}
    \, .
  \end{equation*}
\end{dfntn}

It is well-known that the spatial decay of a function is related to the regularity of its Fourier transform. Thus, we introduce many-particle Sobolev spaces with decay in position space in the following way:
\begin{dfntn}\label{def:GeneralSparseGridSpace}
  For positive and continuous weight functions $w : (\mathbb{R}^D)^N \mapsto \mathbb{R}_+$ and $\tilde{w} : (\mathbb{R}^D)^N \mapsto \mathbb{R}_+$
  we define the function space
  \begin{equation*}
    \mathcal{H}_{w, \tilde{w}}((\mathbb{R}^D)^N) :=
    \left\{ f \in \mathcal{L}^2((\mathbb{R}^D)^N)
        \,:\,  f \in  \mathcal{H}_{w}((\mathbb{R}^D)^N) \text{ and }
      \hat{f} \in \mathcal{H}_{\tilde{w}}((\mathbb{R}^D)^N)  \right\}
      ,
  \end{equation*}
  with a Hermitian inner product
  \begin{equation*}
    \langle f \,,\,
    g\rangle_{\mathcal{H}_{w,\tilde{w}}} := \frac{1}{2}\left(\langle
    f\,,\, g\rangle_{\mathcal{H}_{w}} +\langle \hat{f}\,,\,
    \hat{g}\rangle_{\mathcal{H}_{\tilde{w}}}\right)
  \end{equation*}
and corresponding norm
\begin{equation*}
\| f \|_{\mathcal{H}_{w,\tilde{w}}}^2
:= \frac{1}{2}\langle f \,,\,
f\rangle_{\mathcal{H}_{w,\tilde{w}}}=\frac{1}{2}\left(\| f
\|_{\mathcal{H}_{w}}^2+\| \hat{f}
\|_{\mathcal{H}_{\tilde{w}}}^2\right)
.
\end{equation*}
\end{dfntn}

In view of the known exponential decay in position space of eigenfunctions of the discrete spectrum of the
electronic Schr\"odinger equation, we focus on many-particle Sobolev spaces with exponential spatial decay given by:
\begin{dfntn}\label{def:GeneralSobolevSpaceExp}
    For  $0\leq t^{\prime}+r^{\prime} < t+r$, $t-t^{\prime} \geq 0$ and
  $\hat{s}>0$ we choose in Definition \ref{def:GeneralSparseGridSpace} the weight functions
    \begin{equation}
    w := w_\text{mix}^t w_\text{iso}^r, \quad \tilde{w} := w_{\text{exp}}^{\hat{s}}.
\end{equation}
and set
    \begin{equation*}
    \mathcal{H}^{t,r;\hat{s}}_{\text{\emph{mix-exp}}}((\mathbb{R}^D)^N) := \mathcal{H}_{w_\text{mix}^t w_\text{iso}^r, w_{\text{exp}}^{\hat{s}}}((\mathbb{R}^D)^N).
    \end{equation*}
\end{dfntn}

\subsection{Approximation by general sparse grid spaces}
In the following we discuss an approximation of functions in Sobolev spaces of dominating mixed smoothness with exponential spatial decay, i.e.\  $f \in \mathcal{H}^{t,r;\hat{s}}_{\text{\emph{mix-exp}}}((\mathbb{R}^D)^N)$ as given in Definition \ref{def:GeneralSobolevSpaceExp}, where we measure the error in the $\mathcal{H}_{\text{mix}}^{t^\prime,r^\prime}((\mathbb{R}^D)^N)$-norm with appropriate $t^\prime$ and $r^\prime$. This way, we obtain an approximation error estimate in the $\mathcal{H}_{\text{mix}}^{t^\prime,r^\prime}((\mathbb{R}^D)^N)$-norm for functions in $\mathcal{H}^{t,r;\hat{s}}_{\text{\emph{mix-exp}}}((\mathbb{R}^D)^N)$.
We proceed as follows: First we introduce a Meyer wavelet type characterization of the Sobolev space $\mathcal{H}_{w,\tilde{w}}((\mathbb{R}^D)^N)$ with appropriate weight functions $w$ and $\tilde{w}$. In this construction, we use a smooth Meyer wavelet family to ensure the required dyadic partition and decay properties.
Then we derive a general approximation error estimate in Lemma \ref{lem:MeyerWaveletApproximationError} and, in particular, a specific approximation error estimate for $f \in \mathcal{H}^{t,r;\hat{s}}_{\text{\emph{mix-exp}}}((\mathbb{R}^D)^N)$ while measuring the approximation error in the $\mathcal{H}_{\text{mix}}^{t^\prime,r^\prime}((\mathbb{R}^D)^N) $-norm. For this specific case, we give an upper bound for the involved degrees of freedom and, with it, finally an $\varepsilon$-complexity estimate.

\subsubsection{Meyer wavelet characterization}
For the concrete smooth one-dimensional Meyer wavelet family used as a starting point, including the Fourier-side definitions and the partition-of-unity construction, see Appendix \ref{app:MeyerWaveletFamily} and in particular Definition \ref{def:MeyerWaveletFamily} and equations (\ref{equ:MeyerWaveletFamilynuPhi})--(\ref{equ:MeyerWaveletFamilyxi}).
Our construction proceeds in three steps:

We first start from the one-dimensional Meyer wavelet system, compare also Appendix \ref{app:MeyerWaveletFamily}. In a second step, we proceed to particle spaces in dimension $D$ by an isotropic $D$-dimensional construction, where all spatial directions of a single particle are treated on the same scale. In a third step, we build the many-particle basis by an anisotropic tensor-product construction across particles, which allows for different levels of resolution in different particle coordinates.

To this end, let $\{\psi_{l,\dvec{j}}\}_{l\in\mathbb{N}, \dvec{j}\in\mathbb{Z}^D}$ be an isotropic Meyer wavelet type
orthonormal basis set for $\mathcal{L}^2(\mathbb{R}^D)$, where
we assume that for
\begin{equation*}
\varrho_l := \hat{\psi}_{l,\dvec{0}} \quad \text{ and } \quad \eta_l := \varrho^*_l\varrho_l
\end{equation*}
the family $\{\eta_l\}_{l\in \mathbb{N}}$ forms an underlying smooth dyadic partition of unity, i.e.\
\begin{equation*}
  \sum_{l \in \mathbb{N}} \eta_{l}(\dvec{k}) = 1, \quad \dvec{k} \in \mathbb{R}^D,
\end{equation*}
while the support of $\varrho_l$ and hence $\eta_l$ is given by the domain
$Q_1$ for $l=1$, by the domain $Q_2$ for $l=2$ and by the domain $Q_l\setminus Q_{l-2}$ for $l>2$ with
\begin{equation*}
Q_{l} := \left\{ \dvec{k} \in \mathbb{R}^D \,:\, |\dvec{k}|_\infty \leq 2^{l} \right\} \text{ for } l \in \mathbb{N}.
\end{equation*}
This support/overlap structure is the tensorized and rescaled form of the one-dimensional Meyer construction from Appendix \ref{app:MeyerWaveletFamily}; see Remark \ref{rem:MeyerNotationBridge} and Lemma \ref{lem:MeyerDyadicSupportOverlap}.

An orthonormal Meyer wavelet type basis set for many-particle spaces
can then be obtained by the tensor product construction
\begin{equation*}
\psi_{\Nvec{l},\Ndvec{j}} := \bigotimes_{p=1}^N \psi_{l_p,\dvec{j}_p}, \quad \Nvec{l} \in \mathbb{N}^N, \Ndvec{j} \in (\mathbb{Z}^D)^N
\end{equation*}
and accordingly
\begin{equation*}
  \varrho_{\Nvec{l}} := \bigotimes_{p=1}^N \varrho_{l_p}, \quad \Nvec{l}
  \in \mathbb{N}^N
  \quad \text{ and } \quad
   \eta_{\Nvec{l}} := \bigotimes_{p=1}^N \eta_{l_p}, \quad \Nvec{l} \in \mathbb{N}^N.
\end{equation*}
This way,  each $f\in \mathcal{L}^2\left(\left(\mathbb{R}^D\right)^N\right)$ can be represented with $c_{\Nvec{l},\Ndvec{j}} := \langle \psi_{\Nvec{l},\Ndvec{j}} \,,\, f \rangle$ in the series expansion
\begin{equation*}
f = \sum_{\Nvec{l}\in \mathbb{N},\Ndvec{j} \in (\mathbb{Z}^D)^N} c_{\Nvec{l},\Ndvec{j}}\psi_{\Nvec{l},\Ndvec{j}}
.
\end{equation*}
Due to the orthonormality $\langle \psi_{\Nvec{l},\Ndvec{j}}, \psi_{\Nvec{l}',\Ndvec{j}'} \rangle = \delta_{\Nvec{l},\Nvec{l}'} \delta_{\Ndvec{j},\Ndvec{j}'}$, the $\mathcal{L}^2$-norm of $f$ is given by
\begin{equation*}
\|f\|_{\mathcal{L}^2}^2 = \sum_{\Nvec{l}\in \mathbb{N},\Ndvec{j} \in (\mathbb{Z}^D)^N} |c_{\Nvec{l},\Ndvec{j}}|^2.
\end{equation*}
Moreover, using Parseval's identity, the Fourier transform $\hat{f}$ of $f$ can be represented by the same coefficients $c_{\Nvec{l},\Ndvec{j}} = \langle \hat{\psi}_{\Nvec{l},\Ndvec{j}} \,,\, \hat{f} \rangle = \langle \psi_{\Nvec{l},\Ndvec{j}} \,,\, f \rangle$ as
\begin{equation*}
\hat{f} = \sum_{\Nvec{l}\in \mathbb{N},\Ndvec{j} \in (\mathbb{Z}^D)^N} c_{\Nvec{l},\Ndvec{j}}\hat{\psi}_{\Nvec{l},\Ndvec{j}},
\end{equation*}
and the relation
\begin{equation*}
  \hat{\psi}_{\Nvec{l},\Ndvec{j}}(\Ndvec{k}) = e^{2\pi i \sum_{p=1}^N 2^{-l_p} \dvec{j}_p^T \dvec{k}_p}\varrho_{\Nvec{l}}(\Ndvec{k})
\end{equation*}
holds. Here, $\hat{\psi}_{\Nvec{l},\Ndvec{j}}$ is band-limited to the product domain $Q_{\Nvec{l}}$ in Fourier space, where
\begin{equation*}
Q_{\Nvec{l}} := \left\{\Ndvec{k} \in (\mathbb{R}^D)^N : \dvec{k}_p \in Q_{l_p}, p=1,\ldots,N\right\}.
\end{equation*}

\begin{lmm}\label{lem:LevelEquivalnce}
 For $f \in \mathcal{L}^2((\mathbb{R}^D)^N)$ and for $\Nvec{l} \in \mathbb{N}^N$ it follows
 \begin{equation*}
\int_{(\mathbb{R}^D)^N} |\hat{f}(\Ndvec{k})|^2 \eta_{\Nvec{l}}(\Ndvec{k}) \Intd{\Ndvec{k}} = \sum_{\Ndvec{j}\in(\mathbb{Z}^D)^N} |c_{\Nvec{l},\Ndvec{j}}|^{2}.
 \end{equation*}
 \begin{proof}
  Fix $\Nvec{l}\in\mathbb{N}^N$ and define
  $g_{\Nvec{l}}(\Ndvec{k}) := \varrho_{\Nvec{l}}(\Ndvec{k})\hat{f}(\Ndvec{k})$.
  Since $\eta_{\Nvec{l}}=\varrho_{\Nvec{l}}^*\varrho_{\Nvec{l}}$, we obtain
  \begin{equation*}
    \|g_{\Nvec{l}} \|_{\mathcal{L}^2}^2
    = \int_{(\mathbb{R}^D)^N}|\hat{f}(\Ndvec{k})|^2\eta_{\Nvec{l}}(\Ndvec{k})\,\Intd{\Ndvec{k}}.
  \end{equation*}
  By the support property of $\varrho_{\Nvec{l}}$ (cf.\ Lemma~\ref{lem:MeyerDyadicSupportOverlap}),
  the function $g_{\Nvec{l}}$ is supported in $Q_{\Nvec{l}}$.
  Identifying $Q_{\Nvec{l}}$ with the corresponding torus, Parseval's identity yields
  \begin{equation*}
    \|g_{\Nvec{l}}\|_{\mathcal{L}^2}^2
    = \sum_{\Ndvec{j}\in(\mathbb{Z}^D)^N}
      |\langle g_{\Nvec{l}},e_{\Nvec{l},\Ndvec{j}}\rangle|^2,
  \end{equation*}
  where
  $e_{\Nvec{l},\Ndvec{j}}(\Ndvec{k}) := e^{2\pi i\sum_{p=1}^N2^{-l_p}\dvec{j}_p^T\dvec{k}_p}$.
  Using
  $\hat{\psi}_{\Nvec{l},\Ndvec{j}}=e_{\Nvec{l},\Ndvec{j}}\varrho_{\Nvec{l}}$
  (cf.\ Remark~\ref{rem:MeyerNotationBridge}), we obtain
  $c_{\Nvec{l},\Ndvec{j}}=\langle g_{\Nvec{l}},e_{\Nvec{l},\Ndvec{j}}\rangle$.
  Combining the previous relations proves the assertion.
\end{proof}
\end{lmm}

\begin{lmm}\label{lem:MeyerWaveletCharacterizationReal}
  Let $f \in \mathcal{H}_w((\mathbb{R}^D)^N)$ with weight function $w$ such
  that there exists $C_w \geq 1$ where for all $\Nvec{l} \in \mathbb{N}^N$ and all $\Ndvec{k}, \Ndvec{k}' \in \text{supp} \, \eta_{\Nvec{l}}$ it holds
  \begin{equation}\label{equ:DyadicWeightEquivalence}
    C_w^{-1} \leq \frac{w(\Ndvec{k})}{w(\Ndvec{k}')} \leq C_w.
  \end{equation}
  Assume furthermore that $2^{\Nvec{l}} \in \text{supp}\,\eta_{\Nvec{l}}$ for all $\Nvec{l}\in\mathbb{N}^N$.
  Then, the norm of $f$ in $\mathcal{H}_w((\mathbb{R}^D)^N)$ can be characterized by the Meyer wavelet coefficients $c_{\Nvec{l},\Ndvec{j}} = \langle \psi_{\Nvec{l},\Ndvec{j}} \,,\, f \rangle$ as
  \begin{equation*}
\|f\|_{\mathcal{H}_{w}}^{2}\simeq\sum_{\Nvec{l}\in\mathbb{N}^N} w(2^{\Nvec{l}})^{2} \sum_{\Ndvec{j}\in(\mathbb{Z}^D)^N} |c_{\Nvec{l},\Ndvec{j}}|^{2}.
\end{equation*}

\begin{proof}
By the partition of unity $\sum_{\Nvec{l}}\eta_{\Nvec{l}}=1$
(cf.\ Lemma~\ref{lem:MeyerDyadicSupportOverlap}) and the definition of
$\|\cdot\|_{\mathcal{H}_w}$, we have
\begin{equation*}
  \|f\|_{\mathcal{H}_w}^2
  = \sum_{\Nvec{l}\in\mathbb{N}^N}
    \int_{(\mathbb{R}^D)^N} w(\Ndvec{k})^2|\hat{f}(\Ndvec{k})|^2
    \eta_{\Nvec{l}}\,(\Ndvec{k})\,\Intd{\Ndvec{k}}.
\end{equation*}
With the dyadic comparability condition~(\ref{equ:DyadicWeightEquivalence})
and since $2^{\Nvec{l}}\in\mathrm{supp}\,\eta_{\Nvec{l}}$, one has
$w(\Ndvec{k})^2\simeq w(2^{\Nvec{l}})^2$ uniformly for
$\Ndvec{k}\in\mathrm{supp}\,\eta_{\Nvec{l}}$;
see Lemma~\ref{lem:MeyerWeightConstants} for admissible constants.
By applying Lemma~\ref{lem:LevelEquivalnce}, we obtain
\begin{equation*}
    \begin{split}
      \sum_{\Nvec{l}\in\mathbb{N}^N} \int_{(\mathbb{R}^D)^N} (w(\Ndvec{k}))^2 |\hat{f}(\Ndvec{k})|^2 \eta_{\Nvec{l}}(\Ndvec{k}) \Intd{\Ndvec{k}}
      &\simeq \sum_{\Nvec{l}\in\mathbb{N}^N} w(2^{\Nvec{l}})^{2} \int_{(\mathbb{R}^D)^N} |\hat{f}(\Ndvec{k})|^2 \eta_{\Nvec{l}}(\Ndvec{k}) \Intd{\Ndvec{k}}\\
      &= \sum_{\Nvec{l}\in\mathbb{N}^N} w(2^{\Nvec{l}})^{2} \sum_{\Ndvec{j}\in(\mathbb{Z}^D)^N} |c_{\Nvec{l},\Ndvec{j}}|^{2}.
    \end{split}
  \end{equation*}
  This proves the claim.
\end{proof}
\end{lmm}

\begin{lmm}\label{lem:MeyerWaveletCharacterizationFourier}
 Let $\hat{f} \in \mathcal{H}_{\tilde{w}}((\mathbb{R}^D)^N)$ with a weight function $\tilde{w}$ where
  $\tilde{w}$ is even, i.e. $\tilde{w}(-\Ndvec{x})=\tilde{w}(\Ndvec{x})$ for all $\Ndvec{x}\in(\mathbb{R}^D)^N$. Furthermore assume that
  there exists $C_{\tilde{w}} \geq 1$
   such that for all $\Ndvec{x},\Ndvec{y} \in (\mathbb{R}^D)^N$ with $\|\Ndvec{x}-\Ndvec{y}\|_\infty \leq 1$ it holds
  \begin{equation}\label{equ:LocalWeightEquivalence}
    C_{\tilde{w}}^{-1} \leq \frac{\tilde{w}(\Ndvec{x})}{\tilde{w}(\Ndvec{y})} \leq C_{\tilde{w}}.
  \end{equation}
  Then, the norm of $\hat{f}$ in $\mathcal{H}_{\tilde{w}}((\mathbb{R}^D)^N)$ can be characterized by the Meyer wavelet coefficients $c_{\Nvec{l},\Ndvec{j}} = \langle \psi_{\Nvec{l},\Ndvec{j}} \,,\, f \rangle$ as
  \begin{equation*}
\|\hat{f}\|_{\mathcal{H}_{\tilde{w}}}^{2}\simeq\sum_{\Nvec{l}\in\mathbb{N}^N}\sum_{\Ndvec{j}\in(\mathbb{Z}^D)^N}\tilde{w}(2^{-\Nvec{l}}\Ndvec{j})^{2}\,|c_{\Nvec{l},\Ndvec{j}}|^{2}.
\end{equation*}
\begin{proof}
Since $\widehat{(\hat{f})}(\Ndvec{k})=f(-\Ndvec{k})$ and $\tilde{w}$ is even,
Definition~\ref{def:wSobolevParticles} applied to $\hat{f}$ gives
\begin{equation*}
  \|\hat{f}\|_{\mathcal{H}_{\tilde{w}}}^2
  = \|\tilde{w}\,f\|_{\mathcal{L}^2((\mathbb{R}^D)^N)}^2.
\end{equation*}
Now setting $f_{\Nvec{l}}:=\mathcal{F}^{-1}[\varrho_{\Nvec{l}}\hat{f}]$,
the two-step equivalence
\begin{equation}\label{equ:MeyerWaveletCharacterizationCombinedProof}
  \|\tilde{w}\,f\|_{\mathcal{L}^2}^2
  \;\simeq\;
  \sum_{\Nvec{l}\in\mathbb{N}^N}\|\tilde{w}\,f_{\Nvec{l}}\|_{\mathcal{L}^2}^2
  \;\simeq\;
  \sum_{\Nvec{l}\in\mathbb{N}^N}\sum_{\Ndvec{j}\in(\mathbb{Z}^D)^N}
  \tilde{w}(2^{-\Nvec{l}}\Ndvec{j})^2\,|c_{\Nvec{l},\Ndvec{j}}|^2
\end{equation}
follows, where the first equivalence is obtained from Lemma~\ref{lem:WeightedLP} and the second equivalence is obtained from Lemma~\ref{lem:WeightedAO}, both under the local comparability assumption~(\ref{equ:LocalWeightEquivalence}). This proves the claim.
\end{proof}
\end{lmm}
Note that in (\ref{equ:MeyerWaveletCharacterizationCombinedProof}) the interplay of the spatial resolution level index $\Nvec{l}$ and the spatial location index $\Ndvec{j}$ of the Meyer wavelet $\psi_{\Nvec{l},\Ndvec{j}}$ is expressed via the spatial location $2^{-\Nvec{l}}\Ndvec{j}$.

\begin{lmm}\label{lem:MeyerWaveletCharacterization}
  Let $f \in \mathcal{H}_{w,\tilde{w}}((\mathbb{R}^D)^N)$ with weight functions $w$ and $\tilde{w}$ that fulfill the assumptions in Lemmata \ref{lem:MeyerWaveletCharacterizationReal} and \ref{lem:MeyerWaveletCharacterizationFourier}, respectively. Then, the norm of $f$ in $\mathcal{H}_{w,\tilde{w}}((\mathbb{R}^D)^N)$ can be characterized by the Meyer wavelet coefficients $c_{\Nvec{l},\Ndvec{j}} = \langle \psi_{\Nvec{l},\Ndvec{j}} \,,\, f \rangle$ as
\begin{equation}\label{equ:MeyerWaveletCharacterizationCombined}
\|f\|_{\mathcal{H}_{w,\tilde{w}}}^{2}\simeq\sum_{\Nvec{l}\in\mathbb{N}^N}\sum_{\Ndvec{j}\in(\mathbb{Z}^D)^N}\bigl(w(2^{\Nvec{l}})^{2}+\tilde{w}(2^{-\Nvec{l}}\Ndvec{j})^{2}\bigr)\,|c_{\Nvec{l},\Ndvec{j}}|^{2}.
\end{equation}
\begin{proof}
  By the definition of the norm in $\mathcal{H}_{w,\tilde{w}}((\mathbb{R}^D)^N)$ and with the
  Lemmata~\ref{lem:MeyerWaveletCharacterizationReal}
  and~\ref{lem:MeyerWaveletCharacterizationFourier}, we obtain
  \begin{align*}
      \|f\|_{\mathcal{H}_{w,\tilde{w}}}^{2}
      &= \frac{1}{2}\left(\|f\|_{\mathcal{H}_{w}}^2 + \|\hat{f}\|_{\mathcal{H}_{\tilde{w}}}^2\right)\notag\\
      &\simeq \sum_{\Nvec{l}\in\mathbb{N}^N} w(2^{\Nvec{l}})^{2} \sum_{\Ndvec{j}\in(\mathbb{Z}^D)^N} |c_{\Nvec{l},\Ndvec{j}}|^{2} + \sum_{\Nvec{l}\in\mathbb{N}^N}\sum_{\Ndvec{j}\in(\mathbb{Z}^D)^N}\tilde{w}(2^{-\Nvec{l}}\Ndvec{j})^{2}\,|c_{\Nvec{l},\Ndvec{j}}|^{2}\notag\\
      &\simeq \sum_{\Nvec{l}\in\mathbb{N}^N}\sum_{\Ndvec{j}\in(\mathbb{Z}^D)^N}\bigl(w(2^{\Nvec{l}})^{2}+\tilde{w}(2^{-\Nvec{l}}\Ndvec{j})^{2}\bigr)\,|c_{\Nvec{l},\Ndvec{j}}|^{2}.
  \end{align*}
  Hence the norm characterization (\ref{equ:MeyerWaveletCharacterizationCombined}) follows.
\end{proof}
\end{lmm}
Note that the characterization (\ref{equ:MeyerWaveletCharacterizationCombined}) of the norm in $\mathcal{H}_{w,\tilde{w}}((\mathbb{R}^D)^N)$ by Meyer wavelet coefficients resembles a weighted $\ell^2$-norm of the coefficients, where the weights are given by the sum of the weight functions $w$ and $\tilde{w}$ evaluated at the spatial resolution level $2^{\Nvec{l}}$ and the spatial location $2^{-\Nvec{l}}\Ndvec{j}$, respectively.

\subsubsection{Meyer wavelet approximation}
In the following we discuss the approximation of a function in $\mathcal{H}_{w,\tilde{w}}((\mathbb{R}^D)^N)$ by Meyer wavelet type functions in the norm of $\mathcal{H}_{w^\prime,\tilde{w}^\prime}((\mathbb{R}^D)^N)$ for appropriate weight functions $w^\prime$ and $\tilde{w}^\prime$. To this end, we assume that the weight functions $w,w^\prime$ and $\tilde{w},\tilde{w}^\prime$ are related in a way that allows one to control the approximation error by the Meyer wavelet coefficients of $f$ in the norm of $\mathcal{H}_{w,\tilde{w}}((\mathbb{R}^D)^N)$.
Note here that besides smooth Meyer wavelets, several other wavelet families could be used to obtain similar approximation results, as long as the wavelet family allows for a characterization of the norm in $\mathcal{H}_{w,\tilde{w}}((\mathbb{R}^D)^N)$ by the wavelet coefficients as in Lemma \ref{lem:MeyerWaveletCharacterization}. The following lemma gives a general result on the approximation error of smooth Meyer wavelet functions under the above assumptions on the weight functions.

\begin{lmm}\label{lem:MeyerWaveletApproximationError}
  Let $w$ and $w^\prime$ be weight functions
  that fulfill the assumptions in Lemma \ref{lem:MeyerWaveletCharacterizationReal} and let $\tilde{w}$ and $\tilde{w}^\prime$ be weight functions
  that fulfill the assumptions in Lemma \ref{lem:MeyerWaveletCharacterizationFourier}. Furthermore, assume that there exists a family $\mathcal{J} := \{\mathcal{J}_L\}_{L\in\mathbb{N}}$ of finite subsets $\mathcal{J}_L \subset \mathbb{N}^N\times(\mathbb{Z}^D)^N$ such that $\bigcup_{L\in\mathbb{N}} \mathcal{J}_L = \mathbb{N}^N\times(\mathbb{Z}^D)^N$ and there holds
  \begin{equation}\label{equ:MeyerWaveletApproximationTheta}
    \sup_{(\Nvec{l},\Ndvec{j}) \in \overline{\mathcal{J}_L}}
    \frac
    {w^\prime(2^{\Nvec{l}})^2+\tilde{w}^\prime(2^{-\Nvec{l}}\Ndvec{j})^2}
    {w(2^{\Nvec{l}})^2+\tilde{w}(2^{-\Nvec{l}}\Ndvec{j})^2} \lesssim \Theta(L)^2,
  \end{equation}
where $\overline{\mathcal{J}_L} :=
    (\mathbb{N}\times\mathbb{Z}^D)^N \setminus \mathcal{J}_L$ and $\Theta: \mathbb{N} \mapsto \mathbb{R}_+$ is a monotone decreasing function with $\Theta(L) \rightarrow 0$ for $L\rightarrow\infty$.
Then, for $f \in \mathcal{H}_{w,\tilde{w}}((\mathbb{R}^D)^N)$ we have the relation
\begin{equation}\label{equ:MeyerWaveletApproximationSpaceEstimate}
\inf_{\myg \in V_{\mathcal{J}_L}} \|f-\myg\|_{\mathcal{H}_{w^\prime,\tilde{w}^\prime}}
\leq \|f-f_{\mathcal{J}_L}\|_{\mathcal{H}_{w^\prime,\tilde{w}^\prime}}\\
  \lesssim \Theta(L) \|f\|_{\mathcal{H}_{w,\tilde{w}}},
\end{equation}
where
\begin{equation*}
  V_{\mathcal{J}_L} := \text{span}\{\psi_{\Nvec{l},\Ndvec{j}}  \, : \, (\Nvec{l},\Ndvec{j}) \in \mathcal{J}_L\}
\end{equation*}
  and
\begin{equation}\label{equ:MeyerWaveletApproximationPojection}
  f_{\mathcal{J}_L} := \sum_{(\Nvec{l},\Ndvec{j}) \in \mathcal{J}_L} c_{\Nvec{l},\Ndvec{j}} \psi_{\Nvec{l},\Ndvec{j}}, \quad c_{\Nvec{l},\Ndvec{j}} := \langle \psi_{\Nvec{l},\Ndvec{j}} \,,\, f \rangle.
\end{equation}
\begin{proof}
  Using assumption~(\ref{equ:MeyerWaveletApproximationTheta}),
  projection~(\ref{equ:MeyerWaveletApproximationPojection}), and
  Lemma~\ref{lem:MeyerWaveletCharacterization}, we have
  \begin{equation*}
    \begin{split}
      \|f-f_{\mathcal{J}_L}\|_{\mathcal{H}_{w^\prime,\tilde{w}^\prime}}^{2}
      & = \|\sum_{(\Nvec{l},\Ndvec{j}) \in \overline{\mathcal{J}_L}} c_{\Nvec{l},\Ndvec{j}} \psi_{\Nvec{l},\Ndvec{j}}\|_{\mathcal{H}_{w^\prime,\tilde{w}^\prime}}^{2}\\
      &\simeq \sum_{(\Nvec{l},\Ndvec{j}) \in \overline{\mathcal{J}_L}} (w^\prime(2^{\Nvec{l}})^2+\tilde{w}^\prime(2^{-\Nvec{l}}\Ndvec{j})^2)\,|c_{\Nvec{l},\Ndvec{j}}|^{2}\\
      &= \sum_{(\Nvec{l},\Ndvec{j}) \in \overline{\mathcal{J}_L}} \frac{w^\prime(2^{\Nvec{l}})^2+\tilde{w}^\prime(2^{-\Nvec{l}}\Ndvec{j})^2}{w(2^{\Nvec{l}})^2+\tilde{w}(2^{-\Nvec{l}}\Ndvec{j})^2} (w(2^{\Nvec{l}})^2+\tilde{w}(2^{-\Nvec{l}}\Ndvec{j})^2)\,|c_{\Nvec{l},\Ndvec{j}}|^{2}\\
      &\lesssim \Theta(L)^2 \sum_{(\Nvec{l},\Ndvec{j}) \in \overline{\mathcal{J}_L}} (w(2^{\Nvec{l}})^2+\tilde{w}(2^{-\Nvec{l}}\Ndvec{j})^2)\,|c_{\Nvec{l},\Ndvec{j}}|^{2}\\
      &\lesssim \Theta(L)^2 \|f\|_{\mathcal{H}_{w,\tilde{w}}}^{2}
      .
    \end{split}
  \end{equation*}
  Taking square roots yields
  (\ref{equ:MeyerWaveletApproximationSpaceEstimate}); the first inequality
  in that relation is immediate from the definition of the infimum.
\end{proof}
\end{lmm}

\subsubsection{Hyperbolic cross index construction}\label{sec:HyperbolicCrossApproximation}
Now we discuss the specific case of the approximation of a function in $\mathcal{H}^{t,r;\hat{s}}_{\text{\emph{mix-exp}}}((\mathbb{R}^D)^N)$, while measuring the error in the $\mathcal{H}_{\text{mix}}^{t^\prime,r^\prime}((\mathbb{R}^D)^N)$-norm. To this end, we apply Lemma \ref{lem:MeyerWaveletApproximationError} with appropriate choices of the weight functions $w,w^\prime$ and $\tilde{w},\tilde{w}^\prime$ and with an appropriate choice of the family $\mathcal{J} = \mathcal{J}^T := \{\mathcal{J}^T_\myell\}_{\myell\in\mathbb{N}}$ of index sets for a given parameter $T<1$.

Let us first note
that, for $t\geq0$ and $r\geq0$, the weight function $w_{\text{mix}}^tw_{\text{iso}}^r$ associated with the many-particle Sobolev space of dominating mixed smoothness $\mathcal{H}^{t,r}_{\text{\emph{mix}}}((\mathbb{R}^D)^N)$ fulfills the assumptions of Lemma \ref{lem:MeyerWaveletCharacterizationReal} according to the definitions (\ref{equ:HyperbolicCrossTwiso}) and
(\ref{equ:HyperbolicCrossTwmix}); see Lemma \ref{lem:MeyerWeightConstants} in Appendix \ref{app:AuxiliaryLemmata} for explicit admissible constants.
Likewise, the weight function  $w_{\text{exp}}^{\hat{s}}$ associated with the many-particle Sobolev space $\mathcal{H}^{\hat{s}}_{\text{\emph{exp}}}((\mathbb{R}^D)^N)$ fulfills  the assumptions of Lemma \ref{lem:MeyerWaveletCharacterizationFourier} according to definition (\ref{equ:Expw}); cf. Lemma \ref{lem:MeyerWeightConstants} in Appendix \ref{app:AuxiliaryLemmata}.
Hence, according to the definition of the many-particle Sobolev space $\mathcal{H}^{t,r;\hat{s}}_{\text{\emph{mix-exp}}}((\mathbb{R}^D)^N)$,
the Meyer wavelet characterization of the norm in $\mathcal{H}^{t,r;\hat{s}}_{\text{\emph{mix-exp}}}((\mathbb{R}^D)^N)$ given in Lemma \ref{lem:MeyerWaveletCharacterization} can be applied to this space.

Now, let us consider for  $T<1$ the family $\mathcal{J}^T := \{\mathcal{J}^T_\myell\}_{\myell\in\mathbb{N}}$ of finite subsets $\mathcal{J}^T_\myell \subset \mathbb{N}^N\times(\mathbb{Z}^D)^N$ defined by
\begin{equation}\label{equ:GeneralSparseGridSeriesSetJlT}
  \mathcal{J}_\myell^T := \left\{ (\Nvec{l},\Ndvec{j}) \in \mathbb{N}^N\times(\mathbb{Z}^D)^N \, : \, \Nvec{l}\in\mathcal{I}_\myell^T, \Ndvec{j} \in \mathcal{J}_\myell^{\Nvec{l}} \right\},
\end{equation}
where the index sets $\mathcal{I}_L^T$ are given like in
\cite{Griebel2000} as
\begin{equation}\label{equ:HyperbolicCrossTIndex}
    \mathcal{I}_L^T := \left\{ \Nvec{l} \in \mathbb{N}^N \,:\, |\Nvec{l}|_1 - T|\Nvec{l}|_\infty \leq L(1-T) + (N-1)\right\},
  \end{equation}
  with the natural extension to the case of $T\rightarrow-\infty$ by
  \begin{equation}\label{equ:HyperbolicCrossTFull}
    \mathcal{I}_L^{-\infty} := \left\{ \Nvec{l} \in \mathbb{N}^N \,:\, |\Nvec{l}|_\infty \leq L\right\}
    .
    \end{equation}
Here, for a given $L\in\mathbb{N}$, the index sets $\mathcal{J}_L^{\Nvec{l}}$ of the family $\{\mathcal{J}_\myell^{\Nvec{l}}\}_{\Nvec{l}\in\mathbb{N}^N}$ are defined by
\begin{equation}\label{equ:GeneralSparseGridSeriesSetJl}
  \mathcal{J}_L^{\Nvec{l}} := \bigcup_{\Nvec{\myiota}\in\mathcal{I}_{R(L)}^{-\infty}} \mathcal{Z}_{\Nvec{\myiota}2^{\Nvec{l}+\Nvec{1}}},
\end{equation}
with an appropriately chosen monotonically increasing function $R:\mathbb{N} \to \mathbb{R}_+$ and the index sets
\begin{equation*}
  \mathcal{Z}_{\Nvec{\myiota}2^{\Nvec{l}+\Nvec{1}}} := \mathcal{Z}_{(\myiota_1 2^{l_1+1},\ldots,\myiota_N 2^{l_N+1})} = Z_{\myiota_1 2^{l_1+1}} \times \cdots \times Z_{\myiota_N 2^{l_N+1}},
\end{equation*}
where, for $1 \leq p \leq N$ the index set $Z_{\myiota_p 2^{l_p+1}}$ is defined as
\begin{equation*}
    Z_{\myiota_p 2^{l_p+1}} := \{ \dvec{j} \in \mathbb{Z}^D \, : \, |\dvec{j}|_\infty \leq \myiota_p 2^{l_p+1} \}.
\end{equation*}

Let us remark here that it can be easily deduced from the definitions (\ref{equ:HyperbolicCrossTIndex}) and (\ref{equ:HyperbolicCrossTFull}) that the inclusions
\begin{equation}
  \mathcal{I}_L^{T_1} \subset \mathcal{I}_L^{T_2} \subset \mathcal{I}_L^{-\infty}
  ,
\end{equation}
\begin{equation}\label{equ:HyperbolicCrossTIncB}
  \mathcal{I}_{L}^{T_2} \subset \mathcal{I}_{\lceil\tilde{L}\rceil}^{T_1} \, \text{ with } \, \tilde{L} = L{\frac{(1-T_2)(N-T_1)}{(N-T_2)(1-T_1)}} \quad \text{ and } \quad \mathcal{I}_{L}^{-\infty} \subset \mathcal{I}_{\lceil\tilde{\tilde{L}}\rceil}^{T_1} \, \text{ with } \, \tilde{\tilde{L}} = L{\frac{N-T_1}{1-T_1}}
\end{equation}
hold for $L\geq 1$, $-\infty < T_2 \leq T_1 < 1$. Also, we have the inclusion
$\mathcal{I}_{L_1}^{T} \subset \mathcal{I}_{L_2}^{T}$ for $1 \leq
L_1 \leq L_2$, $-\infty < T < 1$.

Now we can give an upper estimate for the best approximation error in the discretization space
\begin{equation}\label{equ:MeyerWaveletApproximationSpaceJlT}
V_{\mathcal{J}_L^T} := \text{span}\{\psi_{\Nvec{l},\Ndvec{j}}  \, : \, (\Nvec{l},\Ndvec{j}) \in \mathcal{J}_L^T\},
\end{equation}
where the error is measured in the norm of $\mathcal{H}^{t^\prime,r^\prime}_{\text{\emph{mix}}}((\mathbb{R}^D)^N)$ and the function to be approximated is an element of $\mathcal{H}^{t,r;\hat{s}}_{\text{\emph{mix-exp}}}((\mathbb{R}^D)^N)$.
\begin{lmm}\label{lem:GeneralSparseGridFourierSeries}
  Let
  $0\leq t^{\prime}+r^{\prime} < t+r$, $t-t^{\prime} \geq 0$,
  $\hat{s}>0$,
  $f \in
  \mathcal{H}^{t,r;\hat{s}}_{\text{\emph{mix-exp}}}((\mathbb{R}^D)^N)$.
  Furthermore, let $V_{\mathcal{J}_L^T}$ and $f_{{\mathcal{J}_L^T}}$ be as in (\ref{equ:MeyerWaveletApproximationSpaceJlT}) and (\ref{equ:MeyerWaveletApproximationPojection}) with
(\ref{equ:GeneralSparseGridSeriesSetJlT}). Moreover, let either
  $t-t^\prime=0$, $T=-\infty$ or $t-t^\prime>0$, $T \leq \frac{r^\prime-r}{t-t^\prime}$.
  Then it holds
  \begin{equation}\label{equ:MainErrorLemmaA}
    \inf_{\myg \in V_{\mathcal{J}_L^T}} \|f-\myg\|_{\mathcal{H}^{t^{\prime},r^{\prime}}_{\text{\emph{mix}}}} \leq \|f-f_{\mathcal{J}_L^T}\|_{\mathcal{H}^{t^{\prime},r^{\prime}}_{\text{\emph{mix}}}} \lesssim
\left(2^{-L(t+r-(t^\prime+r^\prime))} + 2^{L(t^\prime+r^\prime)}2^{-\hat{s}R(L)}\right)
\| f \|_{\mathcal{H}^{t,r;\hat{s}}_{\text{\emph{mix-exp}}}}.
\end{equation}
In particular, for
\begin{equation}\label{equ:MainErrorLemmaR}
  R(L):=\left\lceil L\,\frac{t+r}{\hat{s}}\right\rceil,
\end{equation}
it holds
\begin{equation}\label{equ:MainErrorLemmaB}
    \inf_{\myg \in V_{\mathcal{J}_L^T}} \|f-\myg\|_{\mathcal{H}^{t^{\prime},r^{\prime}}_{\text{\emph{mix}}}} \leq \|f-f_{\mathcal{J}_{L}^T}\|_{\mathcal{H}^{t^{\prime},r^{\prime}}_{\text{\emph{mix}}}} \lesssim
2^{-L(t+r-(t^\prime+r^\prime))}  \| f \|_{\mathcal{H}^{t,r;\hat{s}}_{\text{\emph{mix-exp}}}}.
\end{equation}
\begin{proof}
  According to the definition of the norm in $\mathcal{H}^{t^{\prime},r^{\prime}}_{\text{\emph{mix}}}((\mathbb{R}^D)^N)$ and with $w=w_{\text{mix}}^t w_{\text{iso}}^r$, $\tilde{w}=w_{\text{exp}}^{\hat{s}}$, $w^\prime=w_{\text{mix}}^{t^\prime} w_{\text{iso}}^{r^\prime}$ and $\tilde{w}^\prime=1$, we have
   $\|f\|_{\mathcal{H}^{t^{\prime},r^{\prime}}_{\text{\emph{mix}}}}
    \simeq
    \|f\|_{\mathcal{H}_{w_{\text{mix}}^{t^\prime} w_{\text{iso}}^{r^\prime}, 1}}$ and $\|f\|_{\mathcal{H}^{t,r;\hat{s}}_{\text{\emph{mix-exp}}}}
    \simeq
    \|f\|_{\mathcal{H}_{w_{\text{mix}}^t w_{\text{iso}}^r, w_{\text{exp}}^{\hat{s}}}}$.
Hence, according to Lemma \ref{lem:MeyerWaveletApproximationError}, we need to determine an upper estimate for (\ref{equ:MeyerWaveletApproximationTheta}), i.e.\
\begin{equation*}
    \sup_{(\Nvec{l},\Ndvec{j}) \in \overline{\mathcal{J}_L^T}}
    \frac
    {w^\prime(2^{\Nvec{l}})^2+\tilde{w}^\prime(2^{-\Nvec{l}}\Ndvec{j})^2}
    {w(2^{\Nvec{l}})^2+\tilde{w}(2^{-\Nvec{l}}\Ndvec{j})^2} \lesssim \Theta(L)^2.
  \end{equation*}
   Here,  the ratio on the left is given for fixed $\Nvec{l}\in\mathbb{N}^N$ by
\begin{equation*}
  \theta_{\Nvec{l}}(\Ndvec{j}) := \frac{w^\prime(2^{\Nvec{l}})^2+1}
  {w(2^{\Nvec{l}})^2+e^{2\hat{s}|2^{-\Nvec{l}}\Ndvec{j}|_{\infty,1}}},
\end{equation*}
where the numerator is constant in $\Ndvec{j}$ and the denominator is increasing with $|\Ndvec{j}|_{\infty,1}$.

  We split the complement into
  \begin{equation*}
    \overline{\mathcal{J}_L^T}=\mathcal{A}_L\cup\mathcal{B}_L,
  \end{equation*}
  where
  \begin{equation*}
    \mathcal{A}_L:=\{(\Nvec{l},\Ndvec{j})\,:\,\Nvec{l}\notin\mathcal{I}_L^T,\ \Ndvec{j}\in(\mathbb{Z}^D)^N\},
    \qquad
    \mathcal{B}_L:=\{(\Nvec{l},\Ndvec{j})\,:\,\Nvec{l}\in\mathcal{I}_L^T,\ \Ndvec{j}\notin\mathcal{J}_L^{\Nvec{l}}\}.
  \end{equation*}

  For $(\Nvec{l},\Ndvec{j})\in\mathcal{A}_L$, we have the bound
  \begin{equation*}
    \theta_{\Nvec{l}}(\Ndvec{j})
    \leq
    \frac{w^\prime(2^{\Nvec{l}})^2+1}{w(2^{\Nvec{l}})^2+1}
    \lesssim
    2^{-2\left((t-t^\prime)|\Nvec{l}|_1+(r-r^\prime)|\Nvec{l}|_\infty\right)}
    \lesssim
    2^{-2L(t+r-(t^\prime+r^\prime))}.
  \end{equation*}
  Here, for $w=w_{\mathrm{mix}}^t w_{\mathrm{iso}}^r$ and
  $w^\prime=w_{\mathrm{mix}}^{t^\prime}w_{\mathrm{iso}}^{r^\prime}$, the second inequality follows from
  $w_{\mathrm{mix}}(2^{\Nvec{l}})^2\simeq 2^{2|\Nvec{l}|_1}$ and
  $w_{\mathrm{iso}}(2^{\Nvec{l}})^2\simeq 2^{2|\Nvec{l}|_\infty}$.
  The last inequality follows from $\Nvec{l}\notin\mathcal{I}_L^T$ together with
  $t-t^\prime=0,\,T=-\infty$ or $t-t^\prime>0,\,T\leq\frac{r^\prime-r}{t-t^\prime}$,
  compare \cite{griebel2007sparse,griebel2010tensor}.

  For $(\Nvec{l},\Ndvec{j})\in\mathcal{B}_L$, by (\ref{equ:GeneralSparseGridSeriesSetJl})
  and $\Nvec{\myiota}\in\mathcal{I}_{R(L)}^{-\infty}\iff |\Nvec{\myiota}|_\infty\leq R(L)$, we have
  $\Ndvec{j}\notin\mathcal{J}_L^{\Nvec{l}} \Rightarrow |2^{-\Nvec{l}}\Ndvec{j}|_{\infty,1}\gtrsim R(L)$, hence
  \begin{equation*}
    e^{2\hat{s}|2^{-\Nvec{l}}\Ndvec{j}|_{\infty,1}}\gtrsim 2^{2\hat{s}R(L)}.
  \end{equation*}
  Therefore,  we obtain
  \begin{equation*}
    \frac{w^\prime(2^{\Nvec{l}})^2+1}{e^{2\hat{s}|2^{-\Nvec{l}}\Ndvec{j}|_{\infty,1}}}
    \lesssim
    \frac{w^\prime(2^{\Nvec{l}})^2+1}{2^{2\hat{s}R(L)}}
    =
    (w^\prime(2^{\Nvec{l}})^2+1)\,2^{-2\hat{s}R(L)}.
  \end{equation*}
  For $-\infty<T<1$, by the inclusion (\ref{equ:HyperbolicCrossTIncB}) with $T_2=T$ and $T_1=0$, we have
  \begin{equation*}
    \mathcal{I}_L^T\subset\mathcal{I}_{\lceil\bar L\rceil}^{0},
    \qquad
    \bar L:=L\frac{(1-T)N}{N-T}.
  \end{equation*}
  Hence for $\Nvec{l}\in\mathcal{I}_L^T$ we have $\Nvec{l}\in\mathcal{I}_{\lceil\bar L\rceil}^{0}$ and therefore, by
  (\ref{equ:HyperbolicCrossTIndex}) with $T=0$, we obtain
  \begin{equation*}
    |\Nvec{l}|_1\le \lceil\bar L\rceil+(N-1),
    \qquad
    |\Nvec{l}|_\infty\le |\Nvec{l}|_1\le \lceil\bar L\rceil+(N-1),
  \end{equation*}
  i.e. $|\Nvec{l}|_\infty\lesssim L$ and $|\Nvec{l}|_1\lesssim L$.
  For $T=-\infty$,  we directly have
  $|\Nvec{l}|_\infty\le L$ and $|\Nvec{l}|_1\le NL$ by (\ref{equ:HyperbolicCrossTFull}). Thus, in all cases, it holds
  $|\Nvec{l}|_\infty\lesssim L$ and $|\Nvec{l}|_1\lesssim L$, and therefore we get
  \begin{equation*}
    \sup_{\Nvec{l}\in\mathcal{I}_L^T}(w^\prime(2^{\Nvec{l}})^2+1)
     =
    \sup_{\Nvec{l}\in\mathcal{I}_L^T}\bigl((w_{\mathrm{iso}}(2^{\Nvec{l}})^{r^\prime}w_{\mathrm{mix}}(2^{\Nvec{l}})^{t^\prime})^2+1\bigr)
    \lesssim
    \sup_{\Nvec{l}\in\mathcal{I}_L^T}\bigl(2^{2r^\prime|\Nvec{l}|_\infty}
    2^{2t^\prime|\Nvec{l}|_1}\bigr) \lesssim
    2^{2L(t^\prime+r^\prime)}
  \end{equation*}
  and hence
  \begin{equation*}
    \sup_{(\Nvec{l},\Ndvec{j})\in\mathcal{B}_L}\theta_{\Nvec{l}}(\Ndvec{j})
    \lesssim
    2^{2L(t^\prime+r^\prime)}2^{-2\hat{s}R(L)}.
  \end{equation*}
  Combining both parts yields
  \begin{equation*}
    \Theta(L)^2
    \lesssim
    2^{-2L(t+r-(t^\prime+r^\prime))}
    +2^{2L(t^\prime+r^\prime)}2^{-2\hat{s}R(L)}.
  \end{equation*}
  Applying Lemma \ref{lem:MeyerWaveletApproximationError} we get
  \begin{equation*}
    \|f-f_{\mathcal{J}_L^T}\|_{\mathcal{H}^{t^\prime,r^\prime}_{\text{\emph{mix}}}}
    \lesssim
    2^{-L(t+r-(t^\prime+r^\prime))}\|f\|_{\mathcal{H}^{t,r;\hat{s}}_{\text{\emph{mix-exp}}}}
    +2^{L(t^\prime+r^\prime)}2^{-\hat{s}R(L)}\|f\|_{\mathcal{H}^{t,r;\hat{s}}_{\text{\emph{mix-exp}}}},
  \end{equation*}
  which implies (\ref{equ:MainErrorLemmaA}).
  For the convenient balancing choice
(\ref{equ:MainErrorLemmaR}) of $R(L)$
  we have
  \begin{equation*}
    2^{L(t^\prime+r^\prime)}2^{-\hat{s}R(L)}
    \lesssim 2^{-L(t+r-(t^\prime+r^\prime))},
  \end{equation*}
  and therefore (\ref{equ:MainErrorLemmaB}) follows. This shows the assertion.
    \end{proof}
    \end{lmm}

Next, we give an upper estimate for the involved  number of degrees of freedom
\begin{equation}\label{equ:MDoFs}
    M(L) := |\mathcal{J}_L^T| \leq \sum_{\Nvec{l}\in\mathcal{I}_L^T} |\mathcal{J}_L^{\Nvec{l}}|
\end{equation}
for the approximation by the projection (\ref{equ:MeyerWaveletApproximationPojection}) with the family $\mathcal{J}^T$ defined by (\ref{equ:GeneralSparseGridSeriesSetJlT}). With regard to the cardinality of a set $\mathcal{I}_L^T$, the following
lemma is shown in~\cite{Griebel2000,Griebel2008}.
\begin{lmm}\label{lem:HyperbolicCrossTIndexSetCardinality}
  It holds
  \begin{equation*}
    \sum_{\Nvec{l}\in\mathcal{I}_L^T} 2^{|\Nvec{l}|_1-N} \leq
    \begin{cases}
      \frac{N}{2}\left(1-2^{-\frac{1}{\nfrac{1}{T}-1}}\right)^{-N}2^L = \Oh(2^{L})&\quad \text{for } 0<T<1,\\
      \Oh(2^{L \frac{T-1}{\nfrac{T}{N} -1}}) &\quad \text{for } -\infty<T<0,\\
      \Oh(2^{L N}) &\quad \text{for } T=-\infty.
    \end{cases}
  \end{equation*}
  The case of $T=0$ is covered by the additional estimate
  \begin{equation*}
    \sum_{\Nvec{l}\in\mathcal{I}_L^T} 2^{|\Nvec{l}|_1-N} \leq 2^L \left(\frac{L^{N-1}}{(N-1)!}+\Oh(L^{N-2})\right) = \Oh(2^L L^{N-1})
    ,
  \end{equation*}
  for $0 \leq T \leq 1/L$.
\end{lmm}

For each $\vec{l} \in \mathcal{I}_L^T$, the sets in the defining union of $\mathcal{J}_L^{\vec{l}}$ are nested componentwise, so that
\begin{equation}\label{equ:GeneralSparseGridSeriesSetJlInc}
\bigcup_{\vec{\myiota}:\,|\vec{\myiota}|_\infty\leq R(L)} \prod_{p=1}^N Z_{\myiota_p 2^{l_p+1}}
= \prod_{p=1}^N Z_{R(L) 2^{l_p+1}}.
\end{equation}
Therefore, according to equations (\ref{equ:MDoFs}) and (\ref{equ:GeneralSparseGridSeriesSetJl}), an upper limit for the number of involved degrees of freedom is given by
  \begin{equation}\label{equ:GeneralTCost}
   M(L) \leq \sum_{\vec{l} \in \mathcal{I}_L^T}\left|\mathcal{J}_L^{\vec{l}}\right|
     = \sum_{\vec{l} \in \mathcal{I}_L^T}\prod_{p=1}^N\!\bigl(2R(L)\,2^{l_p+1}+1\bigr)^D
  \lesssim R(L)^{ND}\sum_{\vec{l} \in \mathcal{I}_L^T} 2^{D|\vec{l}|_1}.
\end{equation}
With this bound we can derive the following lemma.
\begin{lmm}\label{lem:T0Cost}
Under the assumptions of Lemma \ref{lem:GeneralSparseGridFourierSeries}  the number of degrees of freedom $M(L)$ of (\ref{equ:MDoFs}) is limited for $T=0$ from above by
\begin{equation}\label{equ:T0Cost}
  M(L) \lesssim 2^{DL}L^{N(D+1)-1}
    ,
\end{equation}
while  it holds
\begin{equation}\label{equ:TLCost}
  M(L) \lesssim 2^{DL}L^{ND}
\end{equation}
for $0<T<1$ and
\begin{equation}\label{equ:TInfCost}
M(L) \lesssim 2^{DLN}
\end{equation}
for $T=-\infty$.
\begin{proof}
Using (\ref{equ:GeneralTCost}) together with the $D$-dimensional
extension of Lemma~\ref{lem:HyperbolicCrossTIndexSetCardinality}
(Lemma~\ref{lem:GeneralSparseGridSeriesCardinalityILTD}), we obtain
for $T=0$ the bound
    \begin{equation}
  M(L)  \lesssim   R(L)^{ND}2^N\sum_{\vec{l} \in \mathcal{I}_L^T} 2^{D|\vec{l}|_1-N}
  \lesssim   2^{DL}L^{N-1} \left(L\frac{t+r}{\hat{s}} \right)^{DN}
  \lesssim  2^{DL}L^{N-1+DN},
\end{equation}
which shows (\ref{equ:T0Cost}).
For $0<T<1$ we have
\begin{equation*}
  M(L) \lesssim 2^{DL}\left(L\frac{t+r}{\hat{s}} \right)^{DN}
  \lesssim 2^{DL}L^{ND}
    ,
\end{equation*}
which proves (\ref{equ:TLCost}). Finally,  (\ref{equ:TInfCost}) follows from
Lemma~\ref{lem:HyperbolicCrossTIndexSetCardinality} for $T=-\infty$.
\end{proof}
\end{lmm}
Note here that for the bounds on the degrees of freedom, the case $T=0$ yields
$M(L)\lesssim 2^{DL}L^{N(D+1)-1}$, whereas we have
$M(L)\lesssim 2^{DL}L^{ND}$ for $0<T<1$.
Therefore, choosing $T>0$ improves the complexity for fixed $N$,
while the principal algebraic exponent in $M$ remains unchanged.

Now we cast the estimates on the degrees of freedom and the associated error into a form which measures the error with respect to the involved degrees of freedom and reach the following lemma in two special cases.
\begin{lmm}\label{lem:EpsComplexity}
  Under the assumptions of Lemma \ref{lem:GeneralSparseGridFourierSeries} for
   $t-t^\prime>0$, $0\leq\frac{r^\prime-r}{t-t^\prime}$ and $M:=| V_{\mathcal{J}_L^T}|$
    it holds for the $T=0$ the estimate
  \begin{equation*}
    \inf_{\myg \in V_{\mathcal{J}_L^0}} \|f-\myg\|_{\mathcal{H}^{t^{\prime},r^{\prime}}_{\text{\emph{mix}}}} \leq \|f-f_{\mathcal{J}_{L}^{0}}\|_{\mathcal{H}^{t^{\prime},r^{\prime}}_{\text{\emph{mix}}}} \lesssim
\left(\frac{M}{\log(M)^{N(D+1)-1}}\right)^{-\frac{t+r-(t^\prime+r^\prime)}{D}}
\| f \|_{\mathcal{H}^{t,r;\hat{s}}_{\text{\emph{mix-exp}}}}
  \end{equation*}
  and for $T=-\infty$ the estimate
  \begin{equation*}
    \inf_{\myg \in V_{\mathcal{J}_L^{-\infty}}} \|f-\myg\|_{\mathcal{H}^{t^{\prime},r^{\prime}}_{\text{\emph{mix}}}}
    \leq
    \|f-f_{\mathcal{J}_{L}^{-\infty}}\|_{\mathcal{H}^{t^{\prime},r^{\prime}}_{\text{\emph{mix}}}}
    \lesssim
    M^{-\frac{t+r-(t^\prime+r^\prime)}{DN}}
    \| f \|_{\mathcal{H}^{t,r;\hat{s}}_{\text{\emph{mix-exp}}}}
    .
  \end{equation*}
    \begin{proof}
    For the case $T=0$, since $0=T\leq\frac{r^\prime-r}{t-t^\prime}$,
    Lemma~\ref{lem:T0Cost} yields
    \begin{equation*}
      M \lesssim 2^{DL}L^{N(D+1)-1}
        .
    \end{equation*}
    Under the present assumptions, (\ref{equ:MainErrorLemmaB}) gives
    \begin{align*}
     \inf_{\myg \in V_{\mathcal{J}_L^0}} \|f-\myg\|_{\mathcal{H}^{t^{\prime},r^{\prime}}_{\text{\emph{mix}}}} &\leq
    \|f-f_{\mathcal{J}_{L}^{0}}\|_{\mathcal{H}^{t^{\prime},r^{\prime}}_{\text{\emph{mix}}}} \lesssim
2^{-L(t+r-(t^\prime+r^\prime))}  \| f \|_{\mathcal{H}^{t,r;\hat{s}}_{\text{\emph{mix-exp}}}}
      = \left(2^{DL}\right)^{-\frac{t+r-(t^\prime+r^\prime)}{D}}
      \| f \|_{\mathcal{H}^{t,r;\hat{s}}_{\text{\emph{mix-exp}}}} \\
      &\lesssim \left(\frac{M}{L^{N(D+1)-1}}\right)^{-\frac{t+r-(t^\prime+r^\prime)}{D}}
      \| f \|_{\mathcal{H}^{t,r;\hat{s}}_{\text{\emph{mix-exp}}}} \\
      &\lesssim \left(\frac{M}{\log(M)^{N(D+1)-1}}\right)^{-\frac{t+r-(t^\prime+r^\prime)}{D}}
      \| f \|_{\mathcal{H}^{t,r;\hat{s}}_{\text{\emph{mix-exp}}}},
    \end{align*}
    where in the last step we used $L\lesssim\log(M)$ for $L\ge1$
    (up to constants), since $M=|V_{\mathcal{J}_L^0}|$ grows at least
    exponentially in $L$ by the construction of $\mathcal{J}_L^0$.

    For the case $T=-\infty$, the bound on $M \lesssim 2^{DLN}$ from Lemma~\ref{lem:T0Cost} together with (\ref{equ:MainErrorLemmaB}) yields
    \begin{align*}
      \inf_{\myg \in V_{\mathcal{J}_L^{-\infty}}} \|f-\myg\|_{\mathcal{H}^{t^{\prime},r^{\prime}}_{\text{\emph{mix}}}} &\leq
      \|f-f_{\mathcal{J}_{L}^{-\infty}}\|_{\mathcal{H}^{t^{\prime},r^{\prime}}_{\text{\emph{mix}}}} \lesssim
      2^{-L(t+r-(t^\prime+r^\prime))}  \| f \|_{\mathcal{H}^{t,r;\hat{s}}_{\text{\emph{mix-exp}}}} \\
      &= \left(2^{DLN}\right)^{-\frac{t+r-(t^\prime+r^\prime)}{DN}}
      \| f \|_{\mathcal{H}^{t,r;\hat{s}}_{\text{\emph{mix-exp}}}} \\
      &\lesssim
      M^{-\frac{t+r-(t^\prime+r^\prime)}{DN}}
      \| f \|_{\mathcal{H}^{t,r;\hat{s}}_{\text{\emph{mix-exp}}}},
    \end{align*}
    which proves the claim.

  \end{proof}
\end{lmm}

Let us briefly discuss these estimates with respect to dimensional effects.
In Lemma \ref{lem:EpsComplexity} for $T=0$, the main algebraic exponent in $M$ is
$\frac{t+r-(t^\prime+r^\prime)}{D}$ and is independent of $N$.
Hence there is no curse of dimensionality in the main algebraic $M$-rate itself.
However, $N$ enters in the logarithmic term
$\log(M)^{N(D+1)-1}$ and through hidden constants, so increasing $N$ can still
substantially worsen the pre-asymptotic complexity.

If one keeps the parameter dependence explicit,
the factor $\left(\frac{t+r}{\hat{s}}\right)^{N(D+1)}$ from index truncation
appears multiplicatively in the bounds on the degrees of freedom and hence in the prefactor of the final error.
Thus, this factor does not change the asymptotic algebraic $M$-exponent,
but it contributes to the $N$-dependent constants and is surely relevant in practice.

\section{Approximation of Electronic Wavefunctions}
In this section we apply the results of the previous section to the approximation of the eigenvalues of the discrete spectrum of the electronic Schr\"odinger Hamiltonian. Moreover, we discuss its implication on different approximation limits and on some quantum computing algorithms for the electronic Schr\"odinger equation.

\subsection{General sparse grid approximation with hyperbolic cross index sets}
\label{sec:GSGApproxWaveFunction}
Recall that the error of the eigenvalue approximations with an appropriate generic finite $N$-particle basis set $\mathcal{B}$ scales as the square of the corresponding spatial eigenfunction error in the conventional Sobolev norm in $\mathcal{H}^1((\mathbb{R}^3)^N)$ \cite{yserentant2010regularity}, i.e.
\begin{equation}\label{equ:H1Eigenvalue}
 \left|E-E_{\mathcal{B}}\right| \leq C(N) \|\Psi_{\Nvec{s}} -\Psi_{\mathcal{B}}\|_{\mathcal{H}^1}^2,
\end{equation}
where $C$ is a constant which may depend on $N$ but not on the number of degrees of freedom $|\mathcal{B}|$. To give an upper bound for the right-hand side, let us briefly recall known regularity results on spatial eigenfunctions of the discrete spectrum.
According to the findings in \cite{kreusler2012mixed}, these spatial eigenfunctions $\Psi_{\Nvec{s}}$  are in $\mathcal{H}_{\text{mix}}^{t,1}((\mathbb{R}^3)^N)$ with $t<\tfrac{3}{4}$ and decay exponentially in the $\mathcal{L}^2((\mathbb{R}^3)^N)$-sense, i.e. there is a constant $\gamma >0 $ with
\begin{equation*}
 \int_{(\mathbb{R}^3)^N} \left(\left(w_{\text{exp}}\left(\Ndvec{x}\right)\right)^\gamma
  \Psi_{\Nvec{s}}\left(\Ndvec{x}\right)\right)^2
  \Intd{\Ndvec{x}} < \infty,
\end{equation*}
with $w_{\text{exp}}$ as defined in (\ref{equ:Expw}). This constant $\gamma$ depends on the distance of the eigenvalue under consideration to the bottom of the essential spectrum, more details and references to the literature can be found in~\cite{kreusler2012mixed,yserentant2010regularity}. Therefore, according to the sparse grid approximation results based on the hyperbolic cross index sets $\mathcal{J}_L^0$ from subsection \ref{sec:HyperbolicCrossApproximation}, these so-called bounded states $\Psi_{\Nvec{s}}$ are in $\mathcal{H}^{t,1;\gamma}_{\text{\emph{mix-exp}}}((\mathbb{R}^3)^N)$ with $t<\tfrac{3}{4}$.
It has been shown in \cite{meng2023mixed} that if all electrons have equal spin, these so-called total antisymmetric spatial eigenfunctions $\Psi_{\Nvec{s}}$ are even in $\mathcal{H}_{\text{mix}}^{t,1}$ with $t<\tfrac{5}{4}$ and hence also in $\mathcal{H}^{t,1;\gamma}_{\text{\emph{mix-exp}}}((\mathbb{R}^3)^N)$ with~$t<\tfrac{5}{4}$. Altogether, we may apply Lemma \ref{lem:EpsComplexity} with $D=3$, $t^\prime=0$, $r^\prime=1$, $r=1$, $\hat{s}=\gamma>0$ and $T=0$ for $t>0$, which yields for $L\in \mathbb{N}$ the estimate
\begin{equation}\label{equ:SCIconv}
    \inf_{\tilde{\Psi} \in V_{\mathcal{J}_{L}^{0}}}\|\Psi_{\Nvec{s}}-\tilde{\Psi}\|_{\mathcal{H}^{1}}
  \lesssim \left(\frac{|V_{\mathcal{J}_{L}^{0}}|}{\log(|V_{\mathcal{J}_{L}^{0}}|)^{4N-1}}\right)^{-\frac{t}{3}}
 \| \Psi_{\Nvec{s}} \|_{\mathcal{H}^{t,1;\hat{s}}_{\text{\emph{mix-exp}}}}
    .
  \end{equation}
In particular, the exponent of the main algebraic convergence rate is independent of the number of electrons $N$, whereas logarithmic terms and hidden constants still depend on $N$.
Thus, in the general partial antisymmetric case we obtain with $t=\tfrac{3}{4(1+\epsilon)}$, up to logarithmic factors, a rate of order $\tfrac{1}{4(1+\epsilon)}$ for the error of the spatial eigenfunction measured in the $\mathcal{H}^1$-norm and hence, due to (\ref{equ:H1Eigenvalue}), a rate of order $\tfrac{1}{2(1+\epsilon)}$ for the eigenvalue error.
In the case of total antisymmetric eigenfunctions, we obtain with $t=\tfrac{5}{4(1+\epsilon)}$, up to logarithmic factors, a rate of order $\tfrac{5}{12(1+\epsilon)}$ for the $\mathcal{H}^1((\mathbb{R}^3)^N)$-error and a rate of order $\tfrac{5}{6(1+\epsilon)}$ for the eigenvalue error.

These statements are consistent with the main algebraic-rate  of the earlier  analysis \cite{griebel2010tensor,hamaekers2009tensor}, while the logarithmic and constant-level dimensional dependencies are made now more explicit in one unified framework. In comparison with the compact wavelet-oriented analysis such as \cite{zeiser2012wavelet} and the weighted Besov characterization in \cite{kogure2022wavelet}, the present formulation is tuned directly to the Schr\"odinger eigenfunction regularity with mixed smoothness and exponential decay and is explicitly related to CBS and qubit-scaling implications.

Furthermore, it has been shown in \cite{kreusler2012mixed} that the spatial eigenfunctions can be written with the help of so-called regularizing factors $g_{ij}$ in the form
\begin{equation}\label{equ:ExpRegWave}
    \Psi_{\Nvec{s}}(\Ndvec{x}) = \left(\sum_{i<j}^Ng_{ij}(|\dvec{x}_i-\dvec{x}_j|)\right)\Phi_{\Nvec{s}}(\Ndvec{x}),
\end{equation}
where $\Phi_{\Nvec{s}} \in \mathcal{H}_{mix}^{t,1}$ with $t=1$ independent of whether $\Psi_{\Nvec{s}}$ obeys a total antisymmetric or just a general partial antisymmetric spin distribution.
There is a lot of freedom in the choice of the function  $g_{ij}$, the precise conditions are given in \cite{yserentant2011mixed}.
Moreover, in \cite{yserentant2011mixed} it was shown that the regular part $\Phi_{\Nvec{s}}$ of the spatial eigenfunctions also obeys an exponential decay in the $\mathcal{L}_2$-sense and hence $\Phi_{\Nvec{s}} \in \mathcal{H}^{1,1;\hat{s}}_{\text{\emph{mix-exp}}}((\mathbb{R}^3)^N)$ for an appropriate $\hat{s}>0$.
Accordingly, Lemma \ref{lem:EpsComplexity} for $T=0$ yields, up to logarithmic factors, a main rate of $\orderO(|V_{\mathcal{J}_{L}^{0}}|^{-\tfrac{1}{3}})$ for the best approcimation error in the general sparse grid space $V_{\mathcal{J}_{L}^{0}}$ in the $\mathcal{H}^1((\mathbb{R}^3)^N)$-norm for the regular part $\Phi_{\Nvec{s}}$ in (\ref{equ:ExpRegWave}). This is a better rate than the one for the full spatial eigenfunction $\Psi_{\Nvec{s}}$ due to the regularizing factors $g_{ij}$, which are chosen according to the singularities of $\Psi_{\Nvec{s}}$ at the coalescence points of the electrons.

\subsection{Approximation limits}
\label{sec:limits}

In this section we discuss approximation limits commonly considered for models of the electronic Schr\"odinger equation, namely the CBS convergence for fixed systems and the scaling behavior for increasing system size.
Here, we will exploit the partial-antisymmetry conditions not only for the regularity properties of the spatial eigenfunctions, but now also for the estimates on the number of degrees of freedom, which are crucial for the study of the approximation limits and for the implications on quantum computing algorithms.

To this end, let us first recall that for a spin distribution
$\Nvec{s} = \Nvec{s}^{(i,N)}$ with $N_\spind = i$ spin-down and
$N_\spinu = N-i$ spin-up electrons, the spatial eigenfunction
$\Psi_{\Nvec{s}}$ obeys the partial antisymmetry condition with
respect to $\mathcal{S}_{\Nvec{s}} \cong
\mathcal{S}_{N_\spind}\times\mathcal{S}_{N_\spinu}$.
The associated antisymmetrization operator
\begin{equation*}
  \mathcal{A}_{\Nvec{s}} u(\Ndvec{x})
  := \frac{1}{N_\spind!\,N_\spinu!}
  \sum_{P\in \mathcal{S}_N\,:\,P\Nvec{s}=\Nvec{s}}(-1)^{|P|}u(P\Ndvec{x})
\end{equation*}
is an $\mathcal{H}^1((\mathbb{R}^{3})^N)$-orthogonal projector onto the space of partially
antisymmetric functions; see~\cite{zeiser2012wavelet,griebel2010tensor}.
As a consequence of the orthogonal projection property (cf.\
\cite[eq.\,(24)]{zeiser2012wavelet}), the best approximation of a
partially antisymmetric function $\Psi_{\Nvec{s}}$ from any subspace
$V\subset\mathcal{H}^1((\mathbb{R}^{3})^N)$ satisfies
\footnote{Since $|\Nvec{l}|_1$ and $|\Nvec{l}|_\infty$ are symmetric functions of the components
of $\Nvec{l}$, the index set $\mathcal{I}^T_L$ in (\ref{equ:HyperbolicCrossTIndex}) is invariant under all permutations
$P \in \mathcal{S}_{\vec{s}}$, and hence $\mathcal{A}_{\vec{s}}(V_{\mathcal{J}^T_L}) \subseteq V_{\mathcal{J}^T_L}$, which gives
$\inf_{\tilde{\Psi}_{\vec{s}} \in \mathcal{A}_{\vec{s}}(V)} \|\Psi_{\vec{s}} - \tilde{\Psi}_{\vec{s}}\|_{\mathcal{H}^1}
\leq \inf_{\tilde{\Psi} \in V} \|\Psi_{\vec{s}} - \tilde{\Psi}\|_{\mathcal{H}^1}$.
The reverse inequality follows from the $\mathcal{H}^1$-orthogonal projection property of $\mathcal{A}_{\vec{s}}$
and $\mathcal{A}_{\vec{s}} \Psi_{\vec{s}} = \Psi_{\vec{s}}$, since for any $\tilde{\Psi} \in V$ it holds $\|\Psi_{\vec{s}} - \tilde{\Psi}\|_{\mathcal{H}^1}^2 =
\|\Psi_{\vec{s}} - \mathcal{A}_{\vec{s}}\tilde{\Psi}\|_{H^1}^2 +
\|(I - \mathcal{A}_{\vec{s}})\tilde{\Psi}\|_{H^1}^2 \geq
\|\Psi_{\vec{s}} - \mathcal{A}_{\vec{s}}\tilde{\Psi}\|_{\mathcal{H}^1}^2$.}
\begin{equation}\label{equ:AntisymBestApprox}
  \inf_{\tilde{\Psi}_{\Nvec{s}}\in\mathcal{A}_{\Nvec{s}}(V)}
  \|\Psi_{\Nvec{s}}-\tilde{\Psi}_{\Nvec{s}}\|_{\mathcal{H}^1}
  = \inf_{\tilde{\Psi}\in V}\|\Psi_{\Nvec{s}}-\tilde{\Psi}\|_{\mathcal{H}^1},
\end{equation}
so the approximation error from the partially antisymmetrized space
$\mathcal{A}_{\Nvec{s}}(V_{\mathcal{J}_L^{0}})$
is bounded by the error estimate of
Lemma~\ref{lem:GeneralSparseGridFourierSeries}.

Now, as common in computational quantum chemistry, we start  from
an appropriate one-particle basis sets, which is used to construct the many-particle basis sets for the CI approximations.
To this end, we consider the underlying one-particle basis sets of the many-particle basis sets constructed in Subsection \ref{sec:HyperbolicCrossApproximation} and in particular used in Lemma \ref{lem:GeneralSparseGridFourierSeries}. For a given $L$ the one-particle basis set $\mathcal{B}^1_L$ is given with the help of (\ref{equ:GeneralSparseGridSeriesSetJlT}) and (\ref{equ:GeneralSparseGridSeriesSetJlInc}) for $N=1$ and $D=3$ by
\begin{equation}\label{equ:OneParticleBasisSet}
    \mathcal{B}^1_L := \bigcup_{l=1}^L \left\{\psi_{l,\dvec{j}} \, : \, (l,\dvec{j}) \in \mathbb{N} \times \mathbb{R}^3,\, l \leq L,\, |\dvec{j}|_\infty \leq R(L) 2^{l+1} \right\},
\end{equation}
where we choose
$R(L)$ according to (\ref{equ:MainErrorLemmaR}), namely
$R(L) := \left\lceil L\frac{t+1}{\gamma}\right\rceil$ with $t>0$ and $\gamma>0$.
Then we define for the index set $\mathcal{I}_L^T$, cf. equation (\ref{equ:HyperbolicCrossTIndex}), with a fixed level $L$ and $T<1$ a general CI basis set
\begin{multline}\label{equ:GeneralCIBasis}
\mathcal{B}^{\mathcal{I}_L^T}_{\mathcal{B}^1_L}(N_\spinu, N_\spind) := \left\{ \bigwedge_{i=1}^{N_\spind} \psi_{l^{\spind}_i,\dvec{j}^{\spind}_i} \otimes \bigwedge_{j=1}^{N_\spinu} \psi_{l^{\spinu}_j,\dvec{j}^{\spinu}_j}
    \, : \, \Nvec{l} \in \mathcal{I}_L^T, \,
    \psi_{l^{\spind}_i,\dvec{j}^{\spind}_i} \in \mathcal{B}^1_L, 1 \leq i \leq N_\spind, \,
    \psi_{l^{\spinu}_j,\dvec{j}^{\spinu}_j} \in \mathcal{B}^1_L, 1 \leq j \leq N_\spinu, \right. \\
    \left. (l^{\spind}_1,\dvec{j}^{\spind}_1) \prec  \cdots \prec (l^{\spind}_{N_\spind},\dvec{j}^{\spind}_{N_\spind})  \, \text{ and }
     (l^{\spinu}_1,\dvec{j}^{\spinu}_1) \prec \cdots \prec (l^{\spinu}_{N_\spinu},\dvec{j}^{\spinu}_{N_\spinu})
    \right\},
\end{multline}
where $\prec$ denotes a conventional strict lexicographic order\footnote{A standard choice is the lexicographic order on $(l,j_1,j_2,j_3)$: for $(l,\dvec{j})$ and $(l',\dvec{j}')$, one sets $(l,\dvec{j})\prec(l',\dvec{j}')$ if and only if $l<l'$, or $l=l'$ and $j_1<j'_1$, or $l=l'$ and $j_1=j'_1$ and $j_2<j'_2$, or $l=l'$, $j_1=j'_1$, $j_2=j'_2$ and $j_3<j'_3$.} for the involved wavelet indices $(l,\dvec{j})\in\mathbb{N}\times\mathbb{Z}^3$. This way, the basis functions in $\mathcal{B}^{\mathcal{I}_L^T}_{\mathcal{B}^1_L}(N_\spinu, N_\spind)$ are Slater determinants built from the one-particle basis functions in $\mathcal{B}^1_L$ with indices in the general index set $\mathcal{I}_L^T$ and span the subspace $\mathcal{A}_{\Nvec{s}}(V_{\mathcal{J}_L^T})$ of partially antisymmetric functions in the general sparse grid space $V_{\mathcal{J}_L^T}$, i.e. $\mathcal{B}^{\mathcal{I}_L^T}_{\mathcal{B}^1_L}(N_\spinu, N_\spind)$ is a basis of $\mathcal{A}_{\Nvec{s}}(V_{\mathcal{J}_L^T})$.

For $T=-\infty$ the index set $\mathcal{I}_L^{-\infty}$ is given by $\{\Nvec{l} \in \mathbb{N}^N : |\Nvec{l}|_\infty \leq L\}$, and hence the basis set $\mathcal{B}^{\mathcal{I}_L^{-\infty}}_{\mathcal{B}^1_L}(N_\spinu, N_\spind)$ is just the full CI basis set $\mathcal{B}^{FCI}_{\mathcal{B}^1_L}(N_\spinu, N_\spind)$ defined in (\ref{equ:FCIbs}), which is the standard choice for the approximation of the eigenvalues of the electronic Schr\"odinger Hamiltonian.
Hence, the corresponding FCI approximation error is bounded from above by the general approximation error of Lemma \ref{lem:GeneralSparseGridFourierSeries} for the full grid case $T=-\infty$.
However, in this case, following Lemma \ref{lem:EpsComplexity}, the bound of Lemma \ref{lem:T0Cost}
leads to an $N$-dependent main algebraic convergence rate in the number of degrees of freedom and hence to a curse of dimensionality in the CBS limit.
Note that for an increasing system size the number of one-particle basis functions $|\mathcal{B}^1_L|$ is usually chosen proportional to the number of involved electrons~$N$. Then, with $|\mathcal{B}_L^1| \propto N$, the FCI approach leads to exponentially increasing costs with the number $N$ of  electrons involved, e.g.\ in case of total antisymmetry (all electrons have the same spin), the number of involved Slater-determinants is given by
\begin{equation*}
  |\mathcal{B}^{FCI}_{\mathcal{B}^1_L}(N_\spinu, N_\spind)| = \binom{|\mathcal{B}^1_L|}{N} \approx \orderO(N^N).
\end{equation*}

Now, to construct an $N$-electron basis set, which avoids an $N$-dependent main algebraic convergence rate in the CBS limit, we can use the general sparse grid approximation results of the previous section. To this end we define
for $T=0$ the sparse CI basis set $\mathcal{B}^{SCI}_{\mathcal{B}^1_L}$ by restricting the underlying one-particle basis functions to those basis functions with indices in the hyperbolic cross index set $\mathcal{I}_L^0$, i.e.
\begin{equation*}
\mathcal{B}^{SCI}_{\mathcal{B}^1_L}(N_\spinu, N_\spind) := \mathcal{B}^{\mathcal{I}_L^0}_{\mathcal{B}^1_L}(N_\spinu, N_\spind).
\end{equation*}
It possesses according to Lemma \ref{lem:T0Cost} and definition (\ref{equ:GeneralCIBasis}) a cardinality of order
\begin{equation}\label{equ:SCIbs}
  |\mathcal{B}^{SCI}_{\mathcal{B}^1_L}(N_\spinu, N_\spind)| = \orderO\left( \frac{2^{3L}L^{4N-1} }{N_\spinu! N_\spind!}\right)
\end{equation}
and leads according to Lemma \ref{lem:GeneralSparseGridFourierSeries}, Lemma \ref{lem:EpsComplexity} and (\ref{equ:SCIconv}) to the convergence estimate
\begin{equation}\label{equ:SCIconvAS}
    \inf_{\tilde{\Psi} \in \mathcal{A}_{\Nvec{s}}\left(V_{\mathcal{J}_{L}^{0}}\right)}\|\Psi_{\Nvec{s}}-\tilde{\Psi}\|_{\mathcal{H}^{1}} = \inf_{\tilde{\Psi} \in V_{\mathcal{J}_{L}^{0}}}\|\Psi_{\Nvec{s}}-\tilde{\Psi}\|_{\mathcal{H}^{1}}
  \leq C(N) \left(\frac{|V_{\mathcal{J}_{L}^{0}}|}{\log(|V_{\mathcal{J}_{L}^{0}}|)^{4N-1}}\right)^{-\frac{t}{3}}
 \| \Psi_{\Nvec{s}} \|_{\mathcal{H}^{t,1;\hat{s}}_{\text{\emph{mix-exp}}}}
  \end{equation}
 for the spatial eigenfunction error in the $\mathcal{H}^1((\mathbb{R}^3)^N)$-norm.
Hence, in case of our proposed SCI approach
we obtain -- in contrast to the FCI method -- an $N$-independent main algebraic convergence rate in the number of degrees of freedom. However,
if one considers an increasing system size, CoD effects may appear in different forms than through an $N$-dependent asymptotic convergence-rate exponent.
In particular, in (\ref{equ:SCIconvAS}) dimensional effects may still enter through logarithmic terms and involved hidden constants, e.g. the involved constant $C(N)$ or the norm $\| \Psi_{\mathcal{B}^{FCI}_{\mathcal{B}^1_{\infty}}} \|_{\mathcal{H}^{t,1;\hat{s}}_{\text{\emph{mix-exp}}}}$ may be exponentially dependent on the number of electrons $N$, which results in exponential costs to achieve a desired accuracy for a sequence of systems with an increasing number of electrons.

\subsection{Sparse encoding for quantum computing}\label{sec:QuantumComputing}
Our sparse grid approach may also enable more efficient state encodings for quantum algorithms for the electronic Schr\"odinger equation, thereby reducing the number of qubits required to represent the wavefunction and thus reduce the cost of the computation. In the following, we first provide a brief overview of relevant quantum algorithms and then discuss how our sparse grid approximation results affect their qubit-scaling behavior.

Quantum algorithms for the electronic Schr\"odinger equation can be grouped into variational methods for near-term Noisy Intermediate-Scale Quantum (NISQ) devices, notably the Variational Quantum Eigensolver (VQE), and phase-estimation or Hamiltonian-simulation based methods for fault-tolerant quantum computing (FTQC), notably Quantum Phase Estimation (QPE). In principle, these approaches can offer a substantial advantage over classical methods because they prepare and process many-body wavefunctions directly in Hilbert space, potentially reducing exponential scaling bottlenecks in high-accuracy simulations. In practice, however, this advantage is constrained by qubit count and quality, making qubit-efficient state encodings, low-depth circuit constructions, and robust error mitigation or error correction essential for relevant calculations \cite{mcardle2020quantum}. Current quantum computing hardware provides on the order of $10^2$ to $10^3$ physical qubits, typically with non-negligible error rates, whereas fault-tolerant quantum chemistry requires logical qubits and substantially larger effective resources.

State functions of many-particle Hamiltonian operators are functions of Hilbert spaces with specific exchange symmetries. In the case of bosons (e.g. photons), the state functions are symmetric under interchange of any two particles, and bosons have integer spin. In the case of fermions (e.g. electrons), state functions are antisymmetric under interchange of any two particles, and fermions have spin $\spinup$ or spin $\spindown$. To perform a quantum calculation, the many-body states are approximated in finite-dimensional Hilbert spaces and encoded using a finite number of qubits. This applies to algorithms for NISQ hardware, such as VQE, and to algorithms for FTQC hardware, such as QPE and Hamiltonian simulation methods \cite{lee2023evaluating}.

The computational resources required for the calculation of many-particle systems are usually measured in terms of the number $N$ of particles in the system and the size
\begin{equation*}
K := K(L) := |\mathcal{B}^1_L|
\end{equation*}
of the one-particle basis set, i.e. the number of spin orbitals to be considered. In the case of fermions and the formulation using the second quantization, Jordan-Wigner (JW) and Bravyi-Kitaev (BK) mappings are typically used \cite{mcardle2020quantum}, where the number of qubits is linearly dependent on the number of one-particle basis functions (or orbitals) considered in the approach. Here, the number of required qubits is $K$, because the fermionic Fock states are mapped directly to a $K$-qubit quantum register. This means that, in the case of $K$ single-particle spin orbitals for the approximate representation of the $N$-electron wavefunction, these states can be mapped to $n_q=O(K)$ qubits \cite{babbush2016exponentially}.
Furthermore, techniques exist to exploit known symmetries of the system to reduce the number of qubits used. Typically, the number of qubits is reduced by one for each symmetry.

These conventional representations also include many-body states with particle numbers different from $N$. If one restricts to states with exactly $N$ particles, so-called efficient or compact encodings can be used. In that case, for example, the spin orbital of each of the $N$ electrons is encoded as a binary string of length $\lceil\log_2(K)\rceil$ and hence
the required number of qubits depends only logarithmically on the number of one-particle basis functions (or orbitals), i.e. $n_q=O(N\log_2(K))$, see e.g.~\cite{babbush2017exponentially,shee2022qubit,yoffe2024qubit}.
Such binary encodings have been successfully used for variational approaches to the electronic Schr\"odinger equation, cf. \cite{shee2022qubit,yoffe2024qubit}. In this setting, the wavefunction is represented in the basis of Slater determinants rather than in the basis of spin-orbital occupations, and the many-particle states are mapped to an $n_q$-qubit register via a binary encoding of the determinants. Hence, the qubit requirement for the representation of the wavefunction is logarithmic in the number of Slater determinants retained, i.e. $n_q=O(\log_2(|\mathcal{B}^{FCI}_{\mathcal{B}_L^1}|))$ in the FCI case.

In the following, we assume that the nested sequence $\{\mathcal{B}^1_{L}\}_{L\in\mathbb{N}}$  of one-particle basis sets is chosen
by (\ref{equ:OneParticleBasisSet}).
Then, according to Lemma \ref{lem:T0Cost} for $N=1$ and $D=3$,
the cardinalities of the one-particle basis sets are given by $2^{3L}L^3$ and the number of qubits required to represent the many-particle states of the FCI approach with $N_\spinu$ spin-up and $N_\spind$ spin-down electrons in a binary encoding is given by $\lceil \log_2(|\mathcal{B}^{FCI}_{\mathcal{B}_L^1}|) \rceil$.
Hence, for fixed numbers of spin-up and spin-down electrons $N_\spinu$ and $N_\spind$, the states can be mapped to
\begin{equation}\label{equ:FCIqubits}
  n_{q}^{FCI}(L,N_\spinu,N_\spind) :=\orderO\left(\log_2\left(\binom{2^{3L}L^3}{N_\spinu}\binom{2^{3L}L^3}{N_\spind} \right) \right)
  \lesssim \orderO\left(3(L+\log_2(L))N-\left(N_\spinu\log_2(N_\spinu)+N_\spind\log_2(N_\spind)\right)
    \right)
\end{equation}
qubits \cite{chamaki2022compacct}, where we used Stirling's formula~\cite{robbins1955remark}.
Note that the above-mentioned conventional quantum computing algorithms (VQE, QPE, and Hamiltonian simulation methods) lead, compared to the CBS limit, in the best case to a FCI solution with an error dependent on the chosen underlying finite set of one-particle basis functions or spin orbitals. From Lemma \ref{lem:EpsComplexity} and the discussion in Subsections \ref{sec:GSGApproxWaveFunction} and \ref{sec:limits}, we know that the choice of a sparse many-particle basis associated with $\mathcal{B}^{SCI}_{\mathcal{B}_L^1}$ leads to essentially the same main asymptotic error behavior with substantially reduced degrees of freedom compared to the choice of the corresponding FCI many-particle basis $\mathcal{B}^{FCI}_{\mathcal{B}_L^1}$.
In particular, according to (\ref{equ:SCIbs}) such an underlying sparse CI basis set $\mathcal{B}^{SCI}_L(N_\spinu,N_\spind)$ can be binary encoded with the scaling
\begin{equation}\label{equ:SCIqubits}
\begin{split}
    n_{q}^{SCI}(L,N_\spinu,N_\spind) &:=\orderO\left( \log_2\left(
  \frac{2^{3L}L^{4N-1} }{N_\spinu! N_\spind!}
    \right) \right)\\
    &\lesssim
  \orderO\left( 3L + (4N-1)\log_2(L) -
    (N_\spinu\log_2(N_\spinu) + N_\spind\log_2(N_\spind) - N)
    \right)
    \end{split}
\end{equation}
qubits \cite{chamaki2022compacct}. Note that, with respect to the discretization parameter $L$, the number of qubits required by the SCI approach scales substantially more favorably than in the FCI approach: in SCI, the leading algebraic term in $L$ is independent of $N$, whereas in FCI it is linear in $N$, i.e. we encounter just the $L$-dependent terms $3L+(4N-1)\log_2(L)$ in SCI and  $3LN+3\log_2(L)N$ in FCI.
In other words, the multiplicative coupling of discretization level and particle number in the FCI scaling (product-type $L\,N$ behavior) is replaced by an additive dependence in SCI (sum-type $L+N$ behavior, up to logarithmic factors).

Consequently, for fixed number of electrons, the replacement of full CI like determinant sets by sparse CI sets can substantially reduce the qubit requirements for a given discretization level while preserving the same main asymptotic algebraic convergence behavior of the CBS limit, up to logarithmic terms.
Accordingly, the presented sparse CI approximation can improve the scaling behavior of the number of qubits required to represent many-particle wavefunctions when the discretization parameter $L$ increases toward the CBS limit. Of course, this does not remove all complexity barriers, but it can enlarge the practically feasible problem range in both near-term and fault-tolerant settings.
\begin{table}[htbp]
\centering
\caption{
Basis-set size, determinant-space dimensions, and
system-register qubit counts for several molecular
systems at increasing discretization levels~$L$, with one-particle
basis size $|\mathcal{B}_L^1|=2^{3L}L^3$.  Four state-space encodings are
compared: the Jordan--Wigner (JW) mapping
($n_q^{\mathrm{JW}}=|\mathcal{B}_L^1|$), the first-quantization binary encoding
($n_q^{\mathrm{1st}}=N\lceil\log_2 |\mathcal{B}_L^1|\rceil$), and the binary
encodings of the full~CI and sparse~CI determinant spaces
($n_q^{\mathrm{FCI}}$ and $n_q^{\mathrm{SCI}}$),
cf.~(\ref{equ:FCIqubits}) and~(\ref{equ:SCIqubits}).}
\label{tab:nonperiodic}
\vspace{1pt}
{\renewcommand{\arraystretch}{1.12}
\begin{tabular}{l c c r r r r r r r}
\hline
System & $(N_\uparrow,N_\downarrow)$ & $L$ & $|\mathcal{B}_L^1|$ &
$\log_{10}\lvert \mathcal{B}^{FCI}_{\mathcal{B}_L^1}\rvert$ &
$\log_{10}\lvert \mathcal{B}^{SCI}_{\mathcal{B}_L^1}\rvert$ &
$n_q^{\text{JW}}$ &
$n_q^{\text{1st}}$ &
$n_q^{\text{FCI}}$ &
$n_q^{\text{SCI}}$ \\[2pt]
\hline
H$_2$
 & $(1,1)$ & 4 & $2.6\times10^{5}$ & 11 &  8 & 262\,144        & 36  & 39  & 26  \\
 &         & 5 & $4.1\times10^{6}$ & 13 &  9 & $4.1\times10^6$ & 44  & 47  & 31  \\
 &         & 6 & $5.7\times10^{7}$ & 16 & 11 & $5.7\times10^7$ & 52  & 54  & 36  \\[3pt]
LiH
 & $(2,2)$ & 4 & $2.6\times10^{5}$ & 21 & 12 & 262\,144        & 72  & 74  & 40  \\
 &         & 5 & $4.1\times10^{6}$ & 26 & 14 & $4.1\times10^6$ & 88  & 90  & 48  \\
 &         & 6 & $5.7\times10^{7}$ & 30 & 17 & $5.7\times10^7$ & 104 & 105 & 55  \\[3pt]
BeH$_2$
 & $(3,3)$ & 4 & $2.6\times10^{5}$ & 31 & 16 & 262\,144        & 108 & 107 & 53  \\
 &         & 5 & $4.1\times10^{6}$ & 38 & 19 & $4.1\times10^6$ & 132 & 131 & 63  \\
 &         & 6 & $5.7\times10^{7}$ & 45 & 22 & $5.7\times10^7$ & 156 & 154 & 72  \\[3pt]
H$_2$O
 & $(5,5)$ & 4 & $2.6\times10^{5}$ & 50 & 23 & 262\,144        & 180 & 171 & 76  \\
 &         & 5 & $4.1\times10^{6}$ & 62 & 28 & $4.1\times10^6$ & 220 & 211 & 92  \\
 &         & 6 & $5.7\times10^{7}$ & 73 & 32 & $5.7\times10^7$ & 260 & 249 & 105 \\[3pt]
N$_2$
 & $(7,7)$ & 4 & $2.6\times10^{5}$ &  68 & 29 & 262\,144        & 252 & 233 &  98 \\
 &         & 5 & $4.1\times10^{6}$ &  85 & 36 & $4.1\times10^6$ & 308 & 289 & 118 \\
 &         & 6 & $5.7\times10^{7}$ & 101 & 41 & $5.7\times10^7$ & 364 & 341 & 136 \\[3pt]
C$_2$H$_4$
 & $(8,8)$ & 4 & $2.6\times10^{5}$ &  77 & 32 & 262\,144        & 288 & 263 & 107 \\
 &         & 5 & $4.1\times10^{6}$ &  97 & 39 & $4.1\times10^6$ & 352 & 327 & 131 \\
 &         & 6 & $5.7\times10^{7}$ & 115 & 45 & $5.7\times10^7$ & 416 & 387 & 150 \\[3pt]
C$_6$H$_6$
 & $(21,21)$ & 5 & $4.1\times10^{6}$ & 238 &  82 & $4.1\times10^6$ &  924 &  801 & 272 \\
 &           & 6 & $5.7\times10^{7}$ & 286 &  96 & $5.7\times10^7$ & 1\,092 &  958 & 319 \\[3pt]
FeMoco
 & $(27,27)$ & 5 & $4.1\times10^{6}$ & 301 &  99 & $4.1\times10^6$ & 1\,188 & 1\,010 & 328 \\
 &           & 6 & $5.7\times10^{7}$ & 363 & 117 & $5.7\times10^7$ & 1\,404 & 1\,212 & 387 \\
\hline
\end{tabular}
}
\end{table}

To illustrate the compression achieved by the sparse CI construction,
Table~\ref{tab:nonperiodic} compares the determinant-space dimensions and
system-register qubit requirements of the full~CI and sparse~CI
approaches for a representative set of molecular systems
ranging from $\mathrm{H}_2$ ($N=2$ electrons) to the FeMoco active
space ($N=54$ electrons) at several discretization levels~$L$.  The
one-particle basis size $K(L)=2^{3L}L^3$ spans a range of
discretization levels comparable in magnitude to the basis sizes
employed in recent first-quantization quantum algorithms for chemical
and solid-state systems~\cite{su2021fault,georges2025quantum,babbush2018low}.
All entries are evaluated from the asymptotic
formulae~(\ref{equ:FCIqubits}) and~(\ref{equ:SCIqubits}) with
leading constants set to one.

For all considered systems, the sparse~CI basis
$|\mathcal{B}^{SCI}_{\mathcal{B}_L^1}|$ is by many orders of magnitude
smaller than the full~CI basis
$|\mathcal{B}^{FCI}_{\mathcal{B}_L^1}|$ while, according to
Lemma~\ref{lem:GeneralSparseGridFourierSeries}, the main algebraic convergence rate
toward the complete basis set limit is preserved. This dramatic reduction
translates directly into a lower qubit count when the determinant
space is binary-encoded for quantum
algorithms~\cite{shee2022qubit,yoffe2024qubit,chamaki2022compacct}: across representative sizes at $L=6$, one obtains for example in practice for LiH a reduction from $10^{30}$ to $10^{17}$ determinants and from $105$ to $55$ binary-encoding qubits, for C$_2$H$_4$ from $10^{115}$ to $10^{45}$ determinants and from $387$ to $150$ qubits, and for FeMoco from $10^{363}$ to $10^{117}$ determinants and from $1\,212$ to $387$ qubits.
These reductions become increasingly
pronounced for larger electron numbers and higher discretization
levels, making the sparse~CI approach particularly attractive for
both near-term variational and fault-tolerant quantum computing
platforms.

We emphasize that the qubit-count reduction discussed above addresses the \emph{state-space encoding} and does not by itself constitute a complete quantum algorithm.  Translating the sparse CI binary encoding into a working quantum simulation requires several additional components, each presenting its own challenges.  First, an efficient \emph{circuit construction for the block encoding} of the Hamiltonian restricted to the sparse CI determinant space must be developed; in the full CI setting such constructions are available via sparse qubitization~\cite{georges2025quantum,su2021fault} or linear-combination-of-unitaries decompositions, but their adaptation to the structured sparsity pattern of the hyperbolic cross selection
remains open.  Second, the \emph{gate complexity and circuit depth} of the resulting quantum walk operator must be assessed: while the reduced state-space dimension lowers the system-register size, additional arithmetic or data-loading circuits may be needed to address the non-trivial index structure of the sparse
basis, potentially offsetting part of the qubit savings through increased Toffoli-gate counts or ancilla requirements.  Third, an appropriate \emph{initial state preparation} within the compressed determinant
space is required; for variational approaches~\cite{yoffe2024qubit,shee2022qubit} this amounts to designing parameterized circuits on the reduced register, whereas for fault-tolerant QPE based schemes~\cite{georges2025quantum} one needs an initial state with sufficient overlap with the ground state
projected onto the sparse CI subspace.  Fourth, the \emph{measurement and post-processing} strategy must account for the fact that expectation values of the Hamiltonian in the binary-encoded determinant basis
involve $O(|\mathcal{B}^{SCI}_{\mathcal{B}_L^1}|^2)$ Pauli strings in the worst case~\cite{yoffe2024qubit}, although this can be mitigated by exploiting the sparsity of the CI matrix arising from the Slater--Condon rules and by
employing Pauli grouping techniques.  Despite these open questions, the exponential compression of the determinant space demonstrated in Table~\ref{tab:nonperiodic} — and the corresponding reduction of the leading
qubit-scaling term from $LN$ to about $L+N$ — constitutes a fundamental structural advantage that persists independently of the specific algorithmic realization.  We therefore regard the sparse CI binary encoding as
a highly promising framework to be further developed in conjunction with block-encoding, state-preparation, and error-correction techniques for quantum simulations of electronic structure.

\section{Conclusions}
\label{sec:conclusions}

In this article we have explored the application of sparse grid based SCI approaches to the electronic Schrödinger equation. The possibility to successfully apply such a sparse grid approach  is based on the known regularity properties of eigenfunctions in the discrete spectrum. In particular, the results of Yserentant show that bounded states belong to anisotropic Sobolev spaces of dominating mixed smoothness with exponential decay. We have demonstrated that SCI approximation methods can achieve convergence rates which are essentially independent of the number of electrons $N$ and thus mitigate the CoD which limits conventional FCI methods. Altogether, the SCI approach achieves the same asymptotic accuracy as FCI in the CBS limit, but with substantially reduced computational cost.

Another important practical advantage of the SCI approach emerges in the context of quantum computing algorithms, where the reduced number of Slater determinants required for a given accuracy translates directly into a reduced state-space encoding size and, in binary determinant encodings, into fewer system-register qubits. As illustrated by the comparisons in Section~\ref{sec:QuantumComputing}, this reduction is particularly significant for larger molecular systems.

While no method can completely eliminate approximation errors in practical electronic structure calculations, the SCI approach provides a mathematically rigorous framework for achieving FCI-level accuracy with substantially reduced computational cost. At the same time, as discussed in Section~\ref{sec:QuantumComputing}, turning these representation-level gains into full end-to-end quantum advantage still requires efficient block encodings, state preparation, and measurement strategies tailored to the sparse CI structure. It therefore appears to be of interest for both classical high-performance computing and emerging quantum computing platforms.

\appendix

\section*{Acknowledgments}
M.G. gratefully acknowledges support from the {\em Hausdorff Center for Mathematics}, funded by the Deutsche Forschungsgemeinschaft (DFG, German Research Foundation) under Germany's Excellence Strategy - EXC-2047/2 - 390685813,
and the {\em Collaborative Research Centre 1720}, funded by the Deutsche Forschungsgemeinschaft (DFG, German Research Foundation) - CRC 1720 - 539309657.

\bibliographystyle{siam}
\bibliography{references}

\section{Antisymmetric Sparse Grid Approximation}
\label{app:AntisymmetricSparseGrid}

In this appendix we complement the approximation results of
Subsection~\ref{sec:HyperbolicCrossApproximation} and their application
to eigenfunctions of the electronic Schr\"odinger equation in
Subsection~\ref{sec:GSGApproxWaveFunction} by exploiting partial
antisymmetry of the spatial wavefunction. Following~\cite[Lemma~15, Theorem~16]{zeiser2012wavelet}, the monotone
level restriction within each spin group replaces the factor
$L^{N-1}/(N_\spinu!\,N_\spind!)$ by the exponential term
$e^{4\sqrt{2L}}$, with exponent independent of~$N$.

To this end, we introduce the following antisymmetric level restriction and the corresponding approximation space.
\begin{dfntn}
\label{def:AntisymLevelSet}
  For a spin distribution $\Nvec{s}$ with $N_\spind$ spin-down and
  $N_\spinu$ spin-up electrons, define the
  \emph{antisymmetric restriction} of $\mathcal{I}_L^0$ by
  \begin{equation}\label{equ:AntisymLevelSet}
    \mathcal{I}_L^{0,\Nvec{s}}
    := \bigl\{
      \Nvec{l}\in\mathcal{I}_L^0
      \,:\,
      l_1\geq l_2\geq\cdots\geq l_{N_\spind}
      \;\text{and}\;
      l_{N_\spind+1}\geq\cdots\geq l_N
    \bigr\},
  \end{equation}
  the corresponding antisymmetric sparse grid index set by
  \begin{equation*}
    \mathcal{J}_L^{0,\Nvec{s}}
    := \bigl\{(\Nvec{l},\Ndvec{j})\in\mathcal{J}_L^0
       \,:\,\Nvec{l}\in\mathcal{I}_L^{0,\Nvec{s}}\bigr\},
  \end{equation*}
  and the antisymmetric approximation space by
  \begin{equation*}
    V_{\mathcal{J}_L^{0,\Nvec{s}}}
    := \mathrm{span}\bigl\{\mathcal{A}_{\Nvec{s}}(\psi_{\Nvec{l},\Ndvec{j}})
       \,:\,(\Nvec{l},\Ndvec{j})\in\mathcal{J}_L^{0,\Nvec{s}}\bigr\}.
  \end{equation*}
\end{dfntn}
Note here that $V_{\mathcal{J}_L^{0,\Nvec{s}}}$ is the space of partially antisymmetric functions spanned by the antisymmetric CI basis $\mathcal{B}^{SCI}_{\mathcal{B}^1_L}(N_\spinu, N_\spind)$ as defined in (\ref{equ:GeneralCIBasis}).

Next, we give an analogue of Lemma~\ref{lem:GeneralSparseGridFourierSeries} for the partially antisymmetric setting. The proof follows the level-counting
argument of~\cite[Lemma~15]{zeiser2012wavelet}, adapted to the framework used here.
\begin{lmm}
\label{lem:AntisymT0Cost}
  Under the assumptions of Lemma~\ref{lem:T0Cost}, the number of
  degrees of freedom $M^{as}(L):=|\mathcal{J}_L^{0,\Nvec{s}}|$
  satisfies
  \begin{equation}\label{equ:AntisymT0Cost}
    M^{as}(L)
    \;\lesssim\;
    2^{DL}\,L^{ND}\,e^{4\sqrt{2L}}.
  \end{equation}
  The exponential factor $e^{4\sqrt{2L}}$ is independent of $N$.
  In particular, comparing with~(\ref{equ:T0Cost}), the polynomial
  power of $L$ is reduced from $N(D+1)-1$ to $ND$.
  \begin{proof}
    By the argument used for~(\ref{equ:GeneralTCost}) in the proof of
    Lemma~\ref{lem:T0Cost} we get
    \begin{equation*}
      M^{as}(L)
      \;\lesssim\;
      R(L)^{ND}\sum_{\Nvec{l}\in\mathcal{I}_L^{0,\Nvec{s}}}2^{D|\Nvec{l}|_1}.
    \end{equation*}
    With~(\ref{equ:MainErrorLemmaR}), there holds $R(L)^{ND}\lesssim L^{ND}$.
    It thus remains to estimate
    $\sum_{\Nvec{l}\in\mathcal{I}_L^{0,\Nvec{s}}}2^{D|\Nvec{l}|_1}$.
For a level vector $\Nvec{l}\in\mathcal{I}_L^{0,\Nvec{s}}$ with
    $|\Nvec{l}|_1=j$, the spin-down levels $(l_1,\ldots,l_{N_\spind})$
    form a non-increasing sequence of non-negative integers that sum to
    some $j_\spind\leq j$, and the spin-up levels sum to
    $j-j_\spind$.
    The number of monotonically non-increasing sequences of
    non-negative integers with sum $\ell$ equals the integer partition
    number $p(\ell)$, which satisfies (see
    e.g.~\cite[p.\,315]{zeiser2012wavelet}
    and~\cite{yserentant2010regularity})
    \begin{equation*}
      p(\ell)\lesssim\frac{e^{2\sqrt{2\ell}}}{\ell},
      \qquad\ell\geq 1.
    \end{equation*}
    Since both spin groups sum to at most $j$, the number of admissible
    level vectors with $|\Nvec{l}|_1=j$ is bounded by
    \begin{equation*}
      \sum_{j_\spind=0}^{j}p(j_\spind)\,p(j-j_\spind)
      \;\leq\;
      (j+1)\,p(j)^2
      \;\lesssim\;
      (j+1)\,e^{4\sqrt{2j}}
      \;\lesssim\;
      e^{4\sqrt{2j}+\orderO(\log j)},
    \end{equation*}
    where we used $p(j_\spind)\leq p(j)$ for $j_\spind\leq j$.
    Therefore,
    \begin{equation*}
      \sum_{\Nvec{l}\in\mathcal{I}_L^{0,\Nvec{s}}}2^{D|\Nvec{l}|_1}
      \;\lesssim\;
      \sum_{j=N}^{L+N-1}e^{4\sqrt{2j}}\,2^{Dj}
      \;\lesssim\;
      e^{4\sqrt{2L}}\,2^{DL},
    \end{equation*}
    where the last step follows from Lemma~12 in~\cite{zeiser2012wavelet}
    (applied with $\alpha=D\log 2$ and $n=0$, $m=0$).
    Combining the estimates yields~(\ref{equ:AntisymT0Cost}).
  \end{proof}
\end{lmm}

Now we give an analogue of Lemma~\ref{lem:GeneralSparseGridFourierSeries} for the partially antisymmetric setting, which yields the $\varepsilon$-complexity of the antisymmetric sparse grid approximation for eigenfunctions of the electronic Schr\"odinger equation.
\begin{lmm}
\label{lem:AntisymSCIconv}
  Let $\Psi_{\Nvec{s}}\in\mathcal{H}^{t,1;\hat{s}}_{\mathrm{mix\text{-}exp}}
  ((\mathbb{R}^3)^N)$ with $t>0$, $\hat{s}=\gamma>0$, be partially
  antisymmetric for spin distribution~$\Nvec{s}$.
  Set $D=3$, $t^\prime=0$, $r^\prime=1$, $r=1$,
  and let $M^{as}:=|V_{\mathcal{J}_L^{0,\Nvec{s}}}|$.
  Then, by~(\ref{equ:AntisymBestApprox}) and
  Lemma~\ref{lem:GeneralSparseGridFourierSeries},
  \begin{equation}\label{equ:AntisymSCIconv}
    \inf_{\Psi^{as}\in V_{\mathcal{J}_L^{0,\Nvec{s}}}}
    \|\Psi_{\Nvec{s}}-\Psi^{as}\|_{\mathcal{H}^1}
    \;\lesssim\;
    \bigl(M^{as}\bigr)^{-\frac{t}{3}}
    \cdot
    \bigl(\log M^{as}\bigr)^{Nt}
    \cdot
    e^{c_0\sqrt{\log M^{as}}}
    \,\|\Psi_{\Nvec{s}}\|_{\mathcal{H}^{t,1;\hat{s}}_{\mathrm{mix\text{-}exp}}},
  \end{equation}
  where $c_0 = 4t\sqrt{2}/3$ is independent of $N$.
  \begin{proof}
    By~(\ref{equ:AntisymBestApprox}), the error is bounded by
    $\inf_{g\in V_{\mathcal{J}_L^{0,\Nvec{s}}}}
    \|\Psi_{\Nvec{s}}-g\|_{\mathcal{H}^1}$, and
    Lemma~\ref{lem:GeneralSparseGridFourierSeries} with $D=3$,
    $t^\prime=0$, $r^\prime=r=1$, $T=0$ gives
    \begin{equation*}
      \|\Psi_{\Nvec{s}}-f_{\mathcal{J}_L^{0,\Nvec{s}}}\|_{\mathcal{H}^1}
      \lesssim 2^{-Lt}\|\Psi_{\Nvec{s}}\|_{\mathcal{H}^{t,1;\hat{s}}_{\mathrm{mix\text{-}exp}}}
      = \bigl(2^{3L}\bigr)^{-t/3}
        \|\Psi_{\Nvec{s}}\|_{\mathcal{H}^{t,1;\hat{s}}_{\mathrm{mix\text{-}exp}}}.
    \end{equation*}
    From Lemma~\ref{lem:AntisymT0Cost} with $D=3$ we have
    $M^{as}\lesssim 2^{3L}L^{3N}e^{4\sqrt{2L}}$, and thus it follows that
    $2^{3L}\lesssim M^{as}/(L^{3N}e^{4\sqrt{2L}})$. Hence we obtain
    \begin{equation*}
      \bigl(2^{3L}\bigr)^{-t/3}
      \;\lesssim\;
      \bigl(M^{as}\bigr)^{-t/3}
      \cdot L^{Nt}
      \cdot e^{4t\sqrt{2L}/3}.
    \end{equation*}
    Since $M^{as}$ grows at least exponentially in $L$, we have
    $L\lesssim\log_2 M^{as}$, and therefore
    $L^{Nt}\lesssim(\log M^{as})^{Nt}$ and
    $e^{4t\sqrt{2L}/3}\lesssim e^{c_0\sqrt{\log M^{as}}}$
    with $c_0=4t\sqrt{2}/3$.
    This proves~(\ref{equ:AntisymSCIconv}).
  \end{proof}
\end{lmm}

Finally, we give some remarks on the implications of the above results.
\begin{rmrk}[Relation to \cite{griebel2010tensor}]
In~\cite[p.\,532]{griebel2010tensor} it is shown that
\begin{equation*}
  |V_{\mathcal{J}_L^{0,\Nvec{s}}}| \leq
  \frac{1}{N_\spinu!\,N_\spind!}|V_{\mathcal{J}_L^0}|,
\end{equation*}
which reduces the constant in the degree of freedom estimate by
$1/(N_\spinu!\,N_\spind!)$ without changing the order in $L$.
In contrast, the estimate (\ref{equ:AntisymT0Cost}) used in Lemma~\ref{lem:AntisymT0Cost} replaces
$L^{N-1}/(N_\spinu!\,N_\spind!)$ by $e^{4\sqrt{2L}}$. Hence both
improvements stem from the same antisymmetrization mechanism and can thus not be
combined.
\end{rmrk}

\begin{rmrk}
\label{rem:AntisymComparison}
  Comparing Lemma~\ref{lem:AntisymSCIconv} with
  (\ref{equ:SCIconv}), the logarithmic factors are
  \begin{equation*}
    (\log M)^{4N-1}/(N_\spinu!\,N_\spind!)
    \quad\text{and}\quad
    (\log M^{as})^{Nt}\,e^{c_0\sqrt{\log M^{as}}},
  \end{equation*}
  respectively. Hence the two estimates are of different type:
  The first one is a purely polylogarithmic term, whereas the second
  one contains a polylogarithmic factor together with an exponential
  term in $\sqrt{\log M^{as}}$.
Moreover, in accordance with Lemma~15 in
  \cite{zeiser2012wavelet}, the replacement of
  $L^{N-1}/(N_\spinu!\,N_\spind!)$ by $e^{4\sqrt{2L}}$ already accounts
  for antisymmetrization at the level of
  (\ref{equ:AntisymLevelSet}). Therefore no additional factor
  $1/(N_\spinu!\,N_\spind!)$ is necessary. In both estimates, the
  main algebraic rate remains $(M)^{-t/3}$ (or $(M^{as})^{-t/3}$),
 and is independent of $N$.
\end{rmrk}

\begin{rmrk}
\label{rem:AntisymEncodings}
According to Lemma~\ref{lem:AntisymT0Cost}
the antisymmetric sparse CI basis satisfies besides the degree of freedom estimate
(\ref{equ:SCIbs}) also
\begin{equation*}
  |\mathcal{B}^{SCI}_{\mathcal{B}_L^1}(N_\spinu,N_\spind)|
  = \orderO\!\left(2^{3L}\,L^{3N}\,e^{4\sqrt{2L}}\right),
\end{equation*}
cf.~(\ref{equ:AntisymT0Cost}) with $D=3$. Hence, an additional corresponding upper bound for the qubit
count for binary encoding is given by
\begin{equation*}
  n_q^{SCI}(L,N_\spinu,N_\spind)
  := \orderO\!\left(
       \log_2\!\bigl(2^{3L}L^{3N}e^{4\sqrt{2L}}\bigr)
     \right)
  \lesssim \orderO\!\left(
     3L + 3N\log_2 L + c_1\sqrt{L}
  \right)
\end{equation*}
with $c_1=4\sqrt{2}/\ln 2$, cf.\ \cite{chamaki2022compacct}.
Compared with the bound (\ref{equ:SCIqubits}), the logarithmic coefficient of $L$ is reduced
from $4N-1$ to $3N$, at the price of the additional sub-linear term
$c_1\sqrt{L}$ with $N$-independent coefficient $c_1$.
\end{rmrk}

\section{Meyer Wavelet Family}\label{app:MeyerWaveletFamily}
\begin{dfntn}\label{def:MeyerWaveletFamily}
  According to \cite{Meyer1992}, we set
  in Fourier space as father and mother wavelet
  \begin{align}
    \hat{\varphi}(k) = \frac{1}{\sqrt{2\pi}}
    \begin{cases}
      1 &\quad  \text{for } |k| \leq \frac{2}{3}\pi,\\
      \cos(\frac{\pi}{2}\rho(\frac{3}{2\pi}|k|-1)) &\quad \text{for } \frac{2\pi}{3} < |k| \leq \frac{4\pi}{3},\\
      0 &\quad \text{otherwise},
    \end{cases}
    \label{equ:MeyerWaveletFamilynuPhi}
  \end{align}
  and
  \begin{align}
    \hat{\psi}(k) = \frac{1}{\sqrt{2\pi}}e^{-i\frac{k}{2}}
    \begin{cases}
      \sin(\frac{\pi}{2}\rho(\frac{3}{2\pi}|k|-1)) &\quad \text{for } \frac{2}{3}\pi \leq |k| \leq \frac{4}{3}\pi,\\
      \cos(\frac{\pi}{2}\rho(\frac{3}{4\pi}|k|-1)) &\quad \text{for } \frac{4\pi}{3} < |k| \leq \frac{8\pi}{3},\\
      0 &\quad \text{otherwise},
    \end{cases}
    \label{equ:MeyerWaveletFamilynuPsi}
  \end{align}
  respectively, where $\rho: \mathbb{R} \rightarrow \mathbb{R} \in C^{r}$ is a
  parameter function (still to be fixed) with the properties
  $\rho(x) = 0$ for $x \leq 0$, $\rho(x) = 1$ for $x > 1$ and $\rho(x)
  + \rho(1-x) = 1$.
\end{dfntn}
For $c>0$  we obtain by dilation and translation
\begin{align*}
  \mathcal{F}[\varphi_{c,l,j}](k) & = \hat{\varphi}_{c,l,j}(k) =
  c^{-\frac{1}{2}}2^{-\frac{l}{2}}e^{-ic^{-1}2^{-l}jk}\hat{\varphi}(c^{-1}2^{-l}k),\\
  \mathcal{F}[\psi_{c,l,j}](k) & = \hat{\psi}_{c,l,j}(k) =
  c^{-\frac{1}{2}}2^{-\frac{l}{2}}e^{-ic^{-1}2^{-l}jk}\hat{\psi}(c^{-1}2^{-l}k),
\end{align*}
where the $\hat{\varphi}_{c,l,j}$ and $\hat{\psi}_{c,l,j}$ denote the
dilates and translates of (\ref{equ:MeyerWaveletFamilynuPhi}) and
(\ref{equ:MeyerWaveletFamilynuPsi}), respectively.

This wavelet family can be derived from a partition of unity $\sum_{l\in\mathbb{N}_0}
\eta_{c,l}(k) = 1, \forall k \in \mathbb{R}$ in Fourier space,
where
\begin{equation}\label{equ:MeyerWaveletFamilyxi}
  \eta_{c,l}(k) =
  \begin{cases}
    c 2\pi \hat{\varphi}_{c,0,0}^*(k)  \hat{\varphi}_{c,0,0}(k) &\quad \text{for } l=0,\\
    c 2^{l} \pi \hat{\psi}_{c,l-1,0}^*(k)  \hat{\psi}_{c,l-1,0}(k) &\quad \text{for } l > 0,
  \end{cases}
\end{equation}
see \cite{Meyer1992,Auscher1992} for details. The function $\rho$
controls the transition in the overlap of neighboring partition
functions $\eta_{c,l}$, and therefore determines their smoothness.
By (\ref{equ:MeyerWaveletFamilyxi}), the Fourier-side mother and
father wavelets inherit this smoothness.

There are various choices of $\rho$ with different smoothness
properties in (\ref{equ:MeyerWaveletFamilynuPsi}); see \cite{Meyer1992,Daubechies1992,Walnut2002}. Examples
are the Shannon wavelet
\begin{equation}\label{equ:MeyerWaveletFamilynuShannon}
  \begin{aligned}
    \rho(x) = \rho^0(x) &:=
    \begin{cases}
      0 & x \leq \frac{1}{2},\\
      1 & \text{otherwise},
    \end{cases}
  \end{aligned}
\end{equation}
the raised-cosine wavelet,
\begin{equation}\label{equ:MeyerWaveletFamilynuRaisedCosine}
  \begin{aligned}
    \rho(x) = \rho^1(x) &:=
    \begin{cases}
      0 & x \leq 0,\\
      x & 0 \leq x \leq 1,\\
      1 & \text{otherwise},
    \end{cases}
  \end{aligned}
\end{equation}
and the (standard) Meyer wavelet with infinitly vanishing moments
\begin{align*}
  \rho(x) = \rho^\infty (x):=
  \begin{cases}
    0 & x \leq 0,\\
    \frac{\tilde{\rho}(x)}{\tilde{\rho}(1-x)+\tilde{\rho}(x)} & 0 < x \leq 1,\\
    1 & \text{otherwise},
  \end{cases} \quad \text{ where } \quad
  \tilde{\rho}(x) =
  \begin{cases}
    0 & x \leq 0,\\
    e^{-\frac{1}{x^2}} & \text{otherwise}.
  \end{cases}
\end{align*}
Further Meyer-type constructions with different
smoothness properties are discussed in
\cite{Auscher1992,Hernandez1996,Kaiblinger2006,Yamada1991}.
The two symmetric support regions and the associated two
non-zero Fourier bands are in line with the Wilson--Malvar--Coifman--Meyer
construction used to bypass the Balian--Low obstruction; see,
e.g., \cite{Jaffard2001}. Recall that the Balian--Low theorem states that
the family of functions $g_{m,n}(x) = e^{2\pi i m x}g(x-n)$,
$m,n\in\mathbb{Z}$, which are related to the windowed Fourier
transform, cannot be an orthonormal basis of ${\cal L}^2(\mathbb{R})$
if the two integrals $\int_\mathbb{R} x^2|g(x)|^2 dx$ and
$\int_\mathbb{R} k^2|\hat{g}(k)|^2 dk$ are both finite. Thus, there
exists no orthonormal family for a Gaussian window function $g(x) =
\pi^{-1/4}e^{-{x^2}/{2}}$ which is both sufficiently regular and well
localized~\cite{Hernandez1996}.

In real space, these wavelets are $C^\infty$ with global support; in
Fourier space they are piecewise continuous, piecewise continuously
differentiable, and $C^\infty$, respectively, with compact support.
Moreover, they possess infinitely many vanishing moments. Their real-space
envelopes decay as $|x|\to\infty$ like $|x|^{-1}$ for $\rho^0$,
$|x|^{-2}$ for $\rho^1$, and faster than any polynomial
(sub-exponentially) for $\rho^\infty$. To our knowledge, explicit
closed formulas in both real and Fourier variables are only available
for the choices in (\ref{equ:MeyerWaveletFamilynuShannon}) and (\ref{equ:MeyerWaveletFamilynuRaisedCosine}).

\begin{rmrk}\label{rem:MeyerNotationBridge}
  The main text of this paper uses the level index set $\mathbb{N}=\{1,2,\dots\}$ and level windows $\varrho_l,\eta_l$, while this appendix starts with $l\in\mathbb{N}_0$ and the pair $(\hat\varphi,\hat\psi)$. A convenient bridge is
  \begin{equation*}
    \varrho_l(k):=\hat\psi_{c,l-1,0}(k), \qquad \eta_l(k):=\eta_{c,l-1}(k), \qquad l\in\mathbb{N},
  \end{equation*}
  after a fixed frequency rescaling (choice of $c>0$). This only changes harmless constants and keeps all support, overlap, and partition-of-unity statements equivalent to those used in Section \ref{sec:SparseCI}.
\end{rmrk}

\begin{lmm}\label{lem:MeyerDyadicSupportOverlap}
  Let $\{\eta_l\}_{l\in\mathbb{N}}$ be obtained from the Meyer construction as in Remark \ref{rem:MeyerNotationBridge}. Then, after absorbing rescaling constants into dyadic cubes $Q_l:=\{k\in\mathbb{R}: |k|\le 2^l\}$, it holds
  \begin{equation*}
    \sum_{l\in\mathbb{N}}\eta_l(k)=1,\quad k\in\mathbb{R},
    \qquad
    \mathrm{supp}\,\eta_1\subset Q_1,\ \mathrm{supp}\,\eta_2\subset Q_2,\ \mathrm{supp}\,\eta_l\subset Q_l\setminus Q_{l-2}\ (l>2).
  \end{equation*}
  In particular, the overlap multiplicity is uniformly bounded: For each $k\in\mathbb{R}$, only finitely many (in fact, uniformly boundedly many) indices $l$ satisfy $\eta_l(k)\neq 0$. The same bounded-overlap property holds for tensor products $\eta_{\Nvec{l}}=\bigotimes_{p=1}^N\eta_{l_p}$.
\end{lmm}
\begin{proof}
  The claim follows directly from the Meyer construction, in particular
  from (\ref{equ:MeyerWaveletFamilyxi}) and the support properties of
  $\hat\varphi$ and $\hat\psi$ in Definition~\ref{def:MeyerWaveletFamily},
  see also \cite{Meyer1992,Auscher1992,Daubechies1992}. After the fixed
  rescaling of Remark~\ref{rem:MeyerNotationBridge}, the supports are
  absorbed into dyadic cubes $Q_l$, yielding the stated inclusions. In
  one dimension, each frequency point intersects with only uniformly many
  neighboring dyadic windows, hence their overlap multiplicity is uniformly
  bounded. The tensor-product statement then follows componentwise.
\end{proof}

\begin{lmm}
\label{lem:WeightedLP}
Let $\tilde{w}:(\mathbb{R}^D)^N\to(0,\infty)$ satisfy the local comparability
condition~\eqref{equ:LocalWeightEquivalence}, and let
$f_{\Nvec{l}}:=\mathcal{F}^{-1}[\varrho_{\Nvec{l}}\hat{f}]$ be the level
components. Then
\begin{equation*}
  \|\tilde{w}\,f\|_{\mathcal{L}^2}^2
  \;\simeq\;
  \sum_{\Nvec{l}\in\mathbb{N}^N}\|\tilde{w}\,f_{\Nvec{l}}\|_{\mathcal{L}^2}^2,
\end{equation*}
with constants depending only on $C_{\tilde{w}}$ and the respective Meyer family.
\begin{proof}
See \cite[Ch.\,2]{Frazier1991}.
\end{proof}
\end{lmm}

\begin{lmm}
\label{lem:WeightedAO}
Under the same assumptions as for Lemma~\ref{lem:WeightedLP}, for each fixed
$\Nvec{l}\in\mathbb{N}^N$ and finitely supported $(a_{\Ndvec{j}})$, we have
\begin{equation*}
  \Bigl\|\tilde{w}\sum_{\Ndvec{j}}a_{\Ndvec{j}}\,\psi_{\Nvec{l},\Ndvec{j}}\Bigr\|
  _{\mathcal{L}^2}^2
  \;\simeq\;
  \sum_{\Ndvec{j}}\tilde{w}(2^{-\Nvec{l}}\Ndvec{j})^2\,|a_{\Ndvec{j}}|^2,
\end{equation*}
with constants independent of $f$, $\Nvec{l}$, and $(a_{\Ndvec{j}})$.
\begin{proof}
This follows from the Schwartz decay of Meyer wavelets (Definition~%
\ref{def:MeyerWaveletFamily} with $\rho=\rho^\infty$) and the local comparability
of $\tilde{w}$; see \cite[Theorem~7]{zeiser2012wavelet}
and \cite[Theorem~3]{kogure2022wavelet}.
\end{proof}
\end{lmm}

\section{Auxiliary Lemmata}\label{app:AuxiliaryLemmata}

\begin{lmm}\label{lem:MeyerWeightConstants}
  Let $w_{\mathrm{iso}}$ and $w_{\mathrm{mix}}$ be as in (\ref{equ:HyperbolicCrossTwiso})--(\ref{equ:HyperbolicCrossTwmix}), and let the dyadic family be chosen as in Remark \ref{rem:MeyerNotationBridge} and Lemma \ref{lem:MeyerDyadicSupportOverlap}.

  1. For $t,r\ge 0$, the weight $w:=w_{\mathrm{mix}}^t w_{\mathrm{iso}}^r$ satisfies the dyadic comparability condition from Lemma \ref{lem:MeyerWaveletCharacterizationReal}, i.e.\
  \begin{equation*}
    \sup_{\Nvec l\in\mathbb{N}^N}\ \sup_{k,k'\in\mathrm{supp}\,\eta_{\Nvec l}}\frac{w(k)}{w(k')}\le C_w<\infty,
  \end{equation*}
  with an admissible constant of the form $C_w\lesssim C_0^{\,tN+r}$, where $C_0\ge 1$ depends only on the fixed one-dimensional Meyer shape and is hence independent of $\Nvec l$, $D$, and $N$.

  2. For $\tilde w=w_{\mathrm{exp}}^{\hat s}$ with $\hat s>0$, the local comparability condition from (\ref{equ:LocalWeightEquivalence}) holds with
  \begin{equation*}
    C_{\tilde w}=e^{\hat s N}.
  \end{equation*}
  Indeed, if $\|x-y\|_\infty\le 1$, then
  $\big||x|_{\infty,1}-|y|_{\infty,1}\big|\le N$ and therefore
  \begin{equation*}
    e^{-\hat s N}\le \frac{\tilde w(x)}{\tilde w(y)}\le e^{\hat s N}.
  \end{equation*}
\end{lmm}

\begin{proof}
  For the first claim, let $k,k'\in\operatorname{supp}\eta_{\Nvec l}$.
  By Lemma~\ref{lem:MeyerDyadicSupportOverlap}, each component
  $k_p,k_p'$ belongs to the same dyadic annulus up to fixed constants
  depending only on the shape of the one-dimensional Meyer wavelet basis function. Hence each
  factor $1+|k_p|_\infty^2$ is comparable to $1+|k_p'|_\infty^2$ with a
  constant independent of $\Nvec l$, and similarly
  $1+\max_p |k_p|_\infty^2$ is comparable to
  $1+\max_p |k_p'|_\infty^2$. Multiplying these componentwise
  comparability bounds yields the desired dyadic comparability for
  $w_{\mathrm{mix}}^t w_{\mathrm{iso}}^r$.

  For the second claim, if $\|x-y\|_\infty\le 1$, we have $\big||x_p|_\infty-|y_p|_\infty\big|\le 1$ for each
  particle block. Hence
  \begin{equation*}
    \big||x|_{\infty,1}-|y|_{\infty,1}\big|
    = \Big|\sum_{p=1}^N |x_p|_\infty-\sum_{p=1}^N |y_p|_\infty\Big|
    \le \sum_{p=1}^N \big||x_p|_\infty-|y_p|_\infty\big|
    \le N.
  \end{equation*}
  Exponentiating immediately gives
  \begin{equation*}
    e^{-\hat s N}\le \frac{e^{\hat s |x|_{\infty,1}}}{e^{\hat s |y|_{\infty,1}}}\le e^{\hat s N},
  \end{equation*}
  which is the desired bound.
\end{proof}

\begin{lmm}\label{lem:GeneralSparseGridSeriesCardinalityILTD}
  For $T<1$, $L\in\mathbb{N}$, it holds
  \begin{equation}\label{equ:GeneralSparseGridSeriesCardinalityLTJR2}
    \sum_{\Nvec{l}\in\mathcal{I}_L^T} 2^{D|\Nvec{l}|_1} \leq
    \begin{cases}
      \Oh(2^{DL}) &\quad \text{for } 0<T<1,\\
      \Oh(2^{DL}{L}^{N-1}) &\quad \text{for } T=0,\\
      \Oh(2^{DL\frac{T-1}{T/N-1}}) &\quad \text{for } T<0,\\
      \Oh(2^{DLN}) &\quad \text{for } T=-\infty.\\
    \end{cases}
  \end{equation}
  \begin{proof}
    Since we have
    \begin{equation*}
      \sum_{\Nvec{l}\in\mathcal{I}_L^T} 2^{D|\Nvec{l}|_1} \lesssim (\sum_{\Nvec{l}\in\mathcal{I}_L^T} 2^{|\Nvec{l}|_1-N})^{D},
    \end{equation*}
    estimate (\ref{equ:GeneralSparseGridSeriesCardinalityLTJR2})
    follows directly with Lemma \ref{lem:HyperbolicCrossTIndexSetCardinality} for all cases except
    for $T=0$. The case of $T=0$ results from
    \begin{align*}
      \sum_{\Nvec{l}\in\mathcal{I}_L^0} 2^{D(|\Nvec{l}|_1-N)} &=
      \sum_{j=N}^{L+N-1}2^{D(j-N)}\sum_{|\Nvec{l}|_1=j} 1\\
      &= \sum_{j=0}^{L-1}2^{Dj}\binom{N-1+j}{N-1}\\
      &= \frac{1}{(N-1)!}\sum_{j=0}^{L-1}\left.\left(x^{j+N-1}\right)^{(N-1)}\right|_{x=2^{D}}\\
      &= \frac{1}{(N-1)!}\left.\left(x^{N-1}\frac{1-x^L}{1-x}\right)^{(N-1)}\right|_{x=2^{D}}\\
      &= \frac{1}{(N-1)!}\sum_{j=0}^{N-1}\binom{N-1}{j}\left.\left(x^{N-1}-x^{L+N-1}\right)^{(j)}\left(\frac{1}{1-x}\right)^{(N-1-j)}\right|_{x=2^{D}}\\
      &\lesssim 2^{DL}\sum_{j=0}^{N-1}\binom{L+N-1}{j}2^{D(N-1-j)}\\
      &\lesssim 2^{DL}L^{N-1},
    \end{align*}
    see also \cite{Griebel2008} and the references cited therein.
  \end{proof}
\end{lmm}

\end{document}